\documentclass[journal]{IEEEtranTICPS}

\usepackage[hidelinks]{hyperref}
\usepackage{breakurl}

\usepackage[utf8]{inputenc}

\usepackage{amsmath}
\usepackage{amssymb,amsfonts}
\usepackage{amsthm}

\usepackage{xspace}

\usepackage{stmaryrd}

\usepackage{comment}
\usepackage{ifthen}
\usepackage{graphicx}

\usepackage{tabularx}

\usepackage{url}

\usepackage[dvipsnames]{xcolor}
\usepackage[most]{tcolorbox}

\usepackage{array}
\usepackage{longtable}
\usepackage{multirow}
\usepackage{tabularx}

\usepackage{enumitem}

\usepackage{soul,color}
\soulregister\cite7
\soulregister\ref7
\soulregister\pageref7

\newtheorem{theorem}{Theorem}[section]

\newtheorem{definition}{Definition}[section]

\newtheorem{example}{Example}[section]

\newcommand{\ensm}[1]{\ensuremath{#1}\xspace}
\newcommand{\function}[2]{\ensm{\mathsf{#1}(#2)}}
\newcommand{\lrangle}[1]{\ensm{\langle#1\rangle}}
\newcommand{\domain}[1]{\function{dom}{#1}}

\newcommand{\set}[1]{\ensm{\{#1\}}\xspace}
\newcommand{\setbar}{\ensm{\ |\ }\xspace}
\newcommand{\es}{\ensm{\emptyset}}

\newcommand{\maparrow}{\ensm{\mapsto}}

\newcommand{\tree}[2]{
    \ensm {
        \displaystyle
        \frac{
            \begin{array}{c}
            #1
            \end{array}
        }{
            #2
        }
    }
}

\newcommand{\mysf}[2][\displaystyle]{\ensm{{#1\mathsf{#2}}}}
\newcommand{\mysfs}[1]{\mysf[\scriptstyle]{#1}}
\newcommand{\mytsf}[2][\normalfont]{\textsf{#1 #2}}
\newcommand{\mytt}[1]{\ensm{\mathtt{#1}}}

\newcommand{\setdef}[1]{\ensm{\mathcal{#1}}}

\newcommand{\iCPSDL}{\textsf{iCPS-DL}\xspace}

\newcommand{\est}{\participant{est}}

\newcounter{marginc}

\newcommand{\dkcol}[1]{{#1}}

\newcommand{\grayshade}[1]{{\setlength{\fboxsep}{2pt} \fcolorbox{gray!30}{gray!30}{#1}}}

\newcommand\tstrut{\rule{0pt}{2.6ex}}         
\newcommand\bstrut{\rule[-1.4ex]{0pt}{0pt}}   

\newcommand{\iCPS}{\mysf{iCPS}}

\newcommand{\plant}{\ensm{P}} 
\newcommand{\IIoT}{\mysf{IIoT}}

\newcommand{\IT}{\mysf{IT}}

\newcommand{\pstate}[1][]{\ifthenelse{\equal{#1}{}}{\ensm{X}}{\ensm{x_{#1}}}}
\newcommand{\pstated}[1][]{\ifthenelse{\equal{#1}{}}{\ensm{X'}}{\ensm{x'_{#1}}}}

\newcommand{\estate}[1][]{\ifthenelse{\equal{#1}{}}{\ensm{\hat{X}}}{\ensm{\hat{x}_{#1}}}}

\newcommand{\psignal}[1][]{\ifthenelse{\equal{#1}{}}{\ensm{W}}{\ensm{w_{#1}}}}
\newcommand{\psignald}[1][]{\ifthenelse{\equal{#1}{}}{\ensm{W'}}{\ensm{w'_{#1}}}}

\newcommand{\pars}[1][]{\ifthenelse{\equal{#1}{}}{\ensm{Z}}{\ensm{z_{#1}}}}

\newcommand{\property}[1][]{\ensm{\pi_{#1}}}
\newcommand{\Properties}{\ensm{\Pi}}
\newcommand{\properties}[1][]{\ensm{\mathsf{pr}_{#1}}}

\newcommand{\efunction}[1][]{\ensm{\phi_{#1}}}
\newcommand{\emodel}{\ensm{\Phi}}
\newcommand{\estimators}[1][]{\ensm{f_{#1}}}

\newcommand{\Classes}{\ensm{\mathcal{K}}}
\newcommand{\class}[1][]{\ensm{k_{#1}}}

\newcommand{\nvector}[1][]{\ifthenelse{\equal{#1}{}}{\ensm{N}}{\ensm{v_{#1}}}}
\newcommand{\avector}[1][]{\ifthenelse{\equal{#1}{}}{\ensm{\mathrm{A}}}{\ensm{a_{#1}}}}

\newcommand{\MQTT}{\mysf{MQTT}}

\newcommand{\reasoning}[1][]{\ifthenelse{\equal{#1}{}}{\ensm{r}}{\function{\ensm{r}}{#1}}}

\newcommand{\inDomain}{\setdef{O}}

\newcommand{\process}{\ensm{P}} 

\newcommand{\GG}{\ensm{G}} 

\newcommand{\VV}[1][]{\ensm{V^{#1}}}
\newcommand{\EE}[1][]{\ensm{E^{#1}}}

\newcommand{\Sensors}[1][]{\ensm{\setdef{S}^{#1}}}

\newcommand{\Actuators}{\ensm{A}}
\newcommand{\ActuatorClasses}{\ensm{\mathcal{A}}}

\newcommand{\Hardware}[1][]{\ensm{H^{#1}}}
\newcommand{\h}{\mysf{\delta}}

\newcommand{\AG}[1][]{\ensm{F^{#1}}}

\newcommand{\translation}[1][]{\ifthenelse{\equal{#1}{}}{\ensm{\tau}}{\function{\ensm{\tau}}{#1}}}

\newcommand{\statemap}[1][]{\ifthenelse{\equal{#1}{}}{\ensm{\phi}}{\function{\ensm{\phi}}{#1}}}

\newcommand{\egraph}[1][]{\ensm{\mysf{G}_{#1}}}
\newcommand{\eV}[1][]{\ensm{\mysf{V}^{#1}}}
\newcommand{\eE}[1][]{\ensm{\mysf{E}^{#1}}}

\newcommand{\argt}[1][]{\ifthenelse{\equal{#1}{}}{\mysf{argt}}{\function{\mysf{argt}}{#1}}}

\newcommand{\fgraph}[1][]{\ensm{\mysf{F}_{#1}}}

\newcommand{\fmap}[1][]{\ifthenelse{\equal{#1}{}}{\mysf{fmapping}}{\function{\mysf{fmapping}}}{#1}}

\newcommand{\iiotconfiguration}[1][]{\ifthenelse{\equal{#1}{}}{\mysf{conf}}{\function{conf}{#1}}}

\newcommand{\modelmap}[1][]{\ifthenelse{\equal{#1}{}}{\ensm{\mu}}{\function{\ensm{\mu}}{#1}}}
\newcommand{\modelmapd}[1][]{\ifthenelse{\equal{#1}{}}{\ensm{\mu'}}{\function{\ensm{\mu'}}{#1}}}

\newcommand{\contextmap}[1][]{\ifthenelse{\equal{#1}{}}{\ensm{\kappa}}{\function{\ensm{\kappa}}{#1}}}

\newcommand{\rolemap}[1][]{\ifthenelse{\equal{#1}{}}{\ensm{\kappa^{\circ}}}{\function{\kappa^{\circ}}{#1}}}
\newcommand{\rolemapi}[2][]{\ifthenelse{\equal{#1}{}}{\function{\kappa^{\circ}}{#2}}{\function{\kappa_{#1}^{\circ}}{#2}} }

\newcommand{\WDN}{\setdef{WDN}}

\newcommand{\wdnprocess}{\mysf{WDN}}

\newcommand{\reservoir}{\mysf{reservoir}}
\newcommand{\junction}{\mysf{junction}}
\newcommand{\tank}{\mysf{tank}}

\newcommand{\pump}{\mysf{pump}}

\newcommand{\flow}{\mysf{flow}}

\newcommand{\head}{\mysf{head}}

\newcommand{\shape}{\mysf{shape}}

\newcommand{\mass}{\mysf{mass}}
\newcommand{\tankmass}{\mysf{tmass}}
\newcommand{\junctionmass}{\mysf{jmass}}
\newcommand{\energy}{\mysf{energy}}

\newcommand{\tmass}{\participant{t}}

\newcommand{\dev}[1][]{\ensm{\mathsf{dev_{#1}}}}

\newcommand{\cntrl}{\participant{controller}}

\newcommand{\link}{\mysf{link}}

\definecolor{mygreen}{rgb}{0,0.6,0}
\definecolor{myblue}{rgb}{0.4,0.5,0}
\definecolor{mygray}{rgb}{0.5,0.5,0.5}
\definecolor{mymauve}{rgb}{0.58,0,0.82}

\newcommand{\Participants}{\setdef{{\color{mymauve} P}}}
\newcommand{\participant}[1]{\ensm{{\color{mymauve}\mathsf{#1}}}}

\newcommand{\p}[1][]{\participant{p_{#1}}}
\newcommand{\q}[1][]{\participant{q_{#1}}}
\newcommand{\rr}[1][]{\participant{r_{#1}}}

\newcommand{\ap}[1][]{\ifthenelse{\equal{#1}{}}{\mysf{\participant{P}}}{\p[#1]}}
\newcommand{\aq}[1][]{\ifthenelse{\equal{#1}{}}{\mysf{\participant{Q}}}{\q[#1]}}

\newcommand{\ar}[1][]{\ifthenelse{\equal{#1}{}}{\mysf{\participant{R}}}{\rr[#1]}}

\newcommand{\Types}{\setdef{U}}
\newcommand{\type}[1]{\mytt{#1}}

\newcommand{\U}[1][]{\ensm{\mathsf{U_{#1}}}}

\newcommand{\nat}{\type{nat}}
\newcommand{\bool}{\type{bool}}
\newcommand{\ON}{\mytt{ON}}
\newcommand{\OFF}{\mytt{OFF}}

\newcommand{\alphabet}{\ensm{\Sigma}}
\newcommand{\calphabet}{\ensm{\alphabet_{\mysf{c}}}}
\newcommand{\element}[1][]{\ensm{\sigma_{#1}}}

\newcommand{\outputSymbol}{\ensm{!}}
\newcommand{\inputSymbol}{\ensm{?}}

\newcommand{\passSymbol}{\ensm{\shortrightarrow}}

\newcommand{\actOut}[2]{\ensm{{#1}_{\outputSymbol} #2}}
\newcommand{\actInp}[2]{\ensm{{#1}_{\inputSymbol} #2}}

\newcommand{\act}[3]{\ensm{#1\passSymbol#2:#3}}

\newcommand{\sep}{\ensm{.}}

\newcommand{\Local}{\setdef{{\color{olive} T}}}

\newcommand{\local}[1][]{\ensm{{\color{olive}T_{#1}}}}
\newcommand{\locald}[1][]{\ensm{{\color{olive}T_{#1}'}}}

\newcommand{\tor}{\ensm{{\color{olive} +}}}

\newcommand{\Tor}[1]{\ensm{{\color{olive}\textstyle \sum_{#1}}}}

\newcommand{\tloop}{\ensm{\mysf{\color{olive} loop}}}
\newcommand{\trec}[1]{\ensm{\tloop\;\tvar{#1}.}}
\newcommand{\tvar}[1]{\mysf{{\color{olive} #1}}}

\newcommand{\tinact}{\ensm{\mathbf{{\color{olive} end}}}}

\newcommand{\nout}[2]{\actOut{#1}{#2}}
\newcommand{\out}[2]{\ensm{\nout{#1}{#2}\sep}}

\newcommand{\ninp}[2]{\actInp{#1}{#2}}
\newcommand{\inp}[2]{\ensm{\ninp{#1}{#2}\sep}}

\newcommand{\participantsf}{\mysf{p}}
\newcommand{\participants}[1][]{\ifthenelse{\equal{#1}{}}{\participantsf}{\function{\participantsf}{#1}}}

\newcommand{\substff}{\ensm{\Leftarrow}}
\newcommand{\substt}[2]{\ensm{[{#2} \substff {#1}]}}

\newcommand{\by}[1][]{\ensm{\stackrel{#1}{\longrightarrow}\xspace}}

\newcommand{\LSet}{\ensm{\setdef{R}}}
\newcommand{\LL}[1][]{\ensm{S_{#1}}}
\newcommand{\LLd}[1][]{\ensm{S_{#1}'}}

\newcommand{\activef}{\mysf{ap}}
\newcommand{\activep}[1]{\function{\activef}{#1}}

\newcommand{\live}[1]{\function{live}{#1}}
\newcommand{\df}[1]{\function{df}{#1}}

\newcommand{\Glob}{\setdef{{\color{blue} G}}}

\newcommand{\glob}[1][]{\ensm{{\color{blue} G_{#1}}}}
\newcommand{\globd}[1][]{\ensm{{\color{blue} G_{#1}'}}}

\newcommand{\gloop}{\ensm{\mysf{\color{blue} loop}}}

\newcommand{\gor}{\ensm{{\color{blue} +}}}
\newcommand{\Gor}[1]{\ensm{{\color{blue} \textstyle \sum_{#1}}}}

\newcommand{\ginact}{\ensm{\mathbf{{\color{blue} end}}}}
\newcommand{\rec}[1]{\ensm{\gloop\;\var{#1}.}}
\newcommand{\var}[1]{\mysf{{\color{blue} #1}}}

\newcommand{\npass}[3]{\act{#1}{#2}{#3}}
\newcommand{\pass}[3]{\ensm{\npass{#1}{#2}{#3}\sep}}

\newcommand{\proves}{\ensm{\vdash}}

\newcommand{\size}[1][\cdot]{\ensm{|#1|}}

\newcommand{\simple}{{\color{mymauve} \mysf[\scriptstyle]{sim}}}
\newcommand{\WDNsimple}{\mysf{WDN_{\simple}}}

\newcommand{\treesimple}[1][]{\ensm{\mysf{T}_{\simple}^{#1}}}

\newcommand{\globsimple}[1][]{\ensm{\glob^{#1}_{\simple}}}
\newcommand{\LLsimple}[1][]{\ensm{\LL[\simple]^{#1}}}

\newcommand{\controller}[1][]{\participant{c_{\scriptstyle #1}}}
\newcommand{\sensor}[1][]{\participant{s_{\scriptstyle #1}}}

\newcommand{\supervisor}[1][]{\participant{v_{#1}}}

\newcommand{\pumpp}[1][]{\participant{u_{#1}}}

\newcommand{\estp}[1][]{\ensm{\participant{e}_{\scriptstyle \mathsf{#1}}}}

\newcommand{\estpp}[1][]{\ensm{\participant{est}_{\scriptstyle \mathsf{#1}}}}
\newcommand{\junctionpp}[1][]{\ensm{\participant{jun}_{\scriptstyle \mathsf{#1}}}}
\newcommand{\linkpp}[1][]{\ensm{\participant{link}_{\scriptstyle \mathsf{#1}}}}

\newcommand{\levelSensor}[1][]{\participant{s_{h}^{#1}}} 
\newcommand{\flowSensor}[1][]{\participant{s_{f}^{#1}}}

\newcommand{\consumerp}[1][]{\participant{cn_{#1}}}
\newcommand{\producerp}[1][]{\participant{pr_{#1}}}

\newcommand{\headt}{\type{head}}
\newcommand{\flowt}{\type{flow}}
\newcommand{\init}{\type{init}}

\newcommand{\localflow}{\local[\mathsf{f}]}
\newcommand{\localhead}{\local[\mathsf{h}]}
\newcommand{\localest}{\local[\mathsf{e}]}

\newcommand{\localpump}{\local[\mathsf{u}]}

\newcommand{\localjun}{\local[\mathsf{j}]}
\newcommand{\locallink}{\local[\mathsf{l}]}

\usepackage{tikz}

\usetikzlibrary{calc}

\usetikzlibrary{arrows,shapes,automata,backgrounds,petri,chains}
\usetikzlibrary{shapes.symbols}
\usetikzlibrary{shadows}

\definecolor{flowC}{RGB}{0, 166, 148}
\definecolor{headC}{RGB}{220,183,4}

\definecolor{constantC}{RGB}{38,76,90}

\definecolor{junctionC}{RGB}{255, 153, 102}
\definecolor{massC}{RGB}{25, 118, 124}
\definecolor{energyC}{RGB}{156, 45, 7}

\definecolor{darkgreen}{RGB}{0, 90, 0}
\definecolor{darkyellow}{RGB}{139,117,0}

\tikzstyle{common-attributes}=[very thick]

\tikzstyle{pstate}=[regular polygon, regular polygon sides=5, common-attributes, inner sep=2]

\tikzstyle{level}=[draw=headC, fill=white]

\tikzstyle{flow}=[draw=flowC, fill=white]
\tikzstyle{head}=[draw=headC, fill=white]

\tikzstyle{constant}=[draw=black, fill=white]

\tikzstyle{tankhead}=[draw=headC, fill=white]

\tikzstyle{estimator} = [rectangle, common-attributes, inner sep=2.5]
\tikzstyle{jmass} = [draw=junctionC, fill=white]
\tikzstyle{tmass}=[draw=massC, fill=white]
\tikzstyle{energy}=[draw=energyC, fill=white]

\tikzstyle{sensor}=[circle, common-attributes, inner sep=2]
\tikzstyle{actuator}=[isosceles triangle, isosceles triangle apex angle=60, rounded corners=1pt, shape border rotate=270, common-attributes, inner sep=1.5]
\tikzstyle{plc}=[rectangle, rounded corners=2pt, fill=white, dashed, draw=black!40]

\tikzstyle{agent}=[rectangle, draw, common-attributes, inner sep=2]
\tikzstyle{controller}=[rectangle, rotate=45, common-attributes, draw=green, fill=white, inner sep=2]

\tikzstyle{reservoir}=[circle,thick,draw=cyan!75,fill=cyan!20,inner sep=2]
\tikzstyle{pump}=[rectangle,thick,draw=black!75,fill=black!20,inner sep=2.5]
\tikzstyle{valve}=[circle,thick,draw=orange!75,fill=orange!20,inner sep=2]
\tikzstyle{tank}=[circle,thick,draw=blue!75,fill=blue!20,inner sep=2]
\tikzstyle{junction}=[circle,thick,draw=red!75,fill=red!20,inner sep=2]
\tikzstyle{pipe}=[rectangle,thick,draw=green!75,fill=green!20,inner sep=2.5]

\tikzstyle{spoint}=[isosceles triangle, isosceles triangle apex angle=60, rounded corners=1pt, shape border rotate=270, common-attributes, inner sep=1.5]

\tikzstyle{iiot-cloud}=[rounded corners=2pt, fill=black!35!green!20,draw=green]
\tikzstyle{iiot-superv}=[rounded corners=2pt, fill=black!20!yellow!20,draw=yellow]
\tikzstyle{bus}=[line width=0.1mm, orange, double=yellow,double distance=0.25mm]
\tikzstyle{communication}=[rectangle, rounded corners=2pt, draw=blue!50, fill=blue!15]

\newcommand{\pipeGraph}[4] {
    \node[above]            at  (#1, #2 + 2) {#4};
    \draw[dotted]               (#1, #2) rectangle (#1 + 1, #2 + 2);
    \node[constant, pstate]
                            at  (#1 + 0.5, #2 + 0.25) (sh#3) {};
    \node[flow, pstate]
                            at  (#1 + 0.5, #2 + 1.75) (fl#3) {};

    \node[energy, estimator]
                            at  (#1 + 0.5, #2 + 1) (ag#3) {}
                                edge[post] (fl#3)
                                edge[pre] (sh#3);
}

\newcommand{\junctionGraph}[4] {
    \node[above]            at  (#1, #2 + 2) {#4};
    \draw[dotted]               (#1, #2) rectangle (#1 + 1, #2 + 2);

    \node[jmass, estimator]
                            at  (#1 + 0.5, #2 + 1.75) (fl#3) {};
    \node[head, pstate]
                            at  (#1 + 0.5, #2 + 0.25) (pr#3) {};
    \node[flow, pstate]
                            at  (#1 + 0.5, #2 + 1) (dem#3) {}
                                edge[pre] (fl#3);
}

\newcommand{\tankGraph}[4] {
    \node[above]            at  (#1, #2 + 2) {#4};
    \draw[dotted]               (#1, #2) rectangle (#1 + 1, #2 + 2);

    \node[constant, pstate]
                            at  (#1 + 0.5, #2 + 1.75) (sh#3) {};
    \node[tankhead, pstate]
                            at  (#1 + 0.5, #2 + 0.25) (pr#3) {};

    \node[tmass, estimator]
                            at  (#1 + 0.5, #2 + 1) (ag#3) {}
                                edge[pre] (sh#3)
                                edge[post] (pr#3);
}

\newcommand{\pumpGraph}[4] {
    \node[above]             at (#1, #2 + 2) {#4};
    \draw[dotted]               (#1, #2) rectangle (#1 + 1, #2 + 2);

    \node[constant, pstate]
                            at  (#1 + 0.5, #2 + 0.25) (sh#3) {};
    \node[flow, pstate]     at  (#1 + 0.5, #2 + 1.75) (fl#3) {};

    \node[energy, estimator]
                            at  (#1 + 0.5, #2 + 1) (ag#3) {}
                                edge[post] (fl#3)
                                edge[pre] (sh#3);
}

\newcommand{\reservoirGraph}[4] {
    \node[above]        at  (#1, #2 + 1) {#4};
    \draw[dotted]           (#1, #2) rectangle (#1 + 1, #2 + 1);

    \node[head, pstate]
                        at  (#1 + 0.5, #2 + 0.5) (pr#3) {};
}

\tikzstyle{vertex} = [circle, very thick, draw = red, fill=red!20,inner sep = 0, minimum size = 0.2cm]

\newcommand{\drawPump}[5] {
    \node[jun] 
                    at  (#1, #2)  (#4) {};

    \draw           (#1 + 0, #2 + #3) arc (90:390:#3);
    \draw           (#1 + 0, #2 + #3) -- (#1 + #3 + #3/2, #2 + #3) --
                    (#1 + #3 + #3/2, #2 + #3/2) -- (#1 + #3*0.866, #2 + #3/2);
    \node[above]    at  (#1, #2 + #3) {#5};
}

\tikzstyle{jun}=[circle,thick, fill, inner sep=0, minimum size = 0.2cm]
\tikzstyle{cntr}=[circle,thick, fill, inner sep=0,minimum size = 0.2cm]

\usepackage{listings}
\usepackage{textcomp}  
\usepackage{upquote}   

\definecolor{mGreen}{rgb}{0,0.6,0}
\definecolor{mGray}{rgb}{0.5,0.5,0.5}
\definecolor{mPurple}{rgb}{0.58,0,0.82}
\definecolor{backgroundColour}{rgb}{0.95,0.95,0.92}

\lstdefinestyle{hask}{
	commentstyle=\color{magenta},
	numberstyle=\tiny\color{mGray},
	stringstyle=\color{mPurple},
	basicstyle=\scriptsize\ttfamily,
	breakatwhitespace=false,
	breaklines=true,
	captionpos=b,
	keepspaces=true,
	numbers=none,
	showspaces=false,
	showstringspaces=false,
	showtabs=false,      
	tabsize=2,
	keywordstyle=[1]\color{olive},
	keywords=[1]{if, then, else, forall},
	escapeinside={(*}{*)}
}

\lstdefinelanguage{CSP-DL}{
	commentstyle=\color{magenta},
        morecomment=[l]{\#},
	numberstyle=\tiny\color{mGray},
	stringstyle=\color{mPurple},
	basicstyle=\scriptsize\ttfamily,
	breakatwhitespace=false,
	breaklines=true,
	captionpos=b,
	keepspaces=true,
	numbers=none,
	showspaces=false,
	showstringspaces=false,
	showtabs=false,      
	tabsize=2,
	keywordstyle=[1]\color{olive},
	keywords=[1]{property, model, translation, physical, actuator, sensor, device, conn, sense, control, actuate, using, estimate},
        keywordstyle=[2]\color{magenta},
        keywords=[2]{paradigm, process, knowledge, base, local, global, translate, traverse, configure, compose, project},
	escapeinside={(*}{*)}
}

\newcommand{\iCPSDLinline}[1]{\lstinline[language=CSP-DL,basicstyle=\ttfamily\small]{#1}}

\lstdefinelanguage{iCPSDL}{
	commentstyle=\color{magenta},
        morecomment=[l]{\#},
	numberstyle=\tiny\color{mGray},
	stringstyle=\color{mPurple},
        basicstyle=\scriptsize\ttfamily,
	breakatwhitespace=false,
	breaklines=true,
	captionpos=b,
	keepspaces=true,
	numbers=none,
	showspaces=false,
	showstringspaces=false,
	showtabs=false,      
	tabsize=2,
	keywordstyle=[1]\color{olive},
	keywords=[1]{property, model, translation, physical, actuator, sensor, device, conn, sense, control, actuate, using, estimate, or},
        keywordstyle=[2]\color{blue},
        keywords=[2]{domain, process, repository, local, global, translate, traverse, configure, compose, project},
	escapeinside={(*}{*)}
}

\newcommand{\iCPSDLinl}[1]{\lstinline[language=iCPSDL,basicstyle=\ttfamily\small]{#1}}

\lstdefinelanguage{antlr4}{
	commentstyle=\color{magenta},
        morecomment=[l]{\#},
	numberstyle=\tiny\color{mGray},
	stringstyle=\color{mPurple},
        basicstyle=\scriptsize\ttfamily,
	breakatwhitespace=false,
	breaklines=true,
	captionpos=b,
	keepspaces=true,
	numbers=none,
        upquote=true,
	showspaces=false,
	showstringspaces=false,
	showtabs=false,      
	tabsize=2,
	escapeinside={(*}{*)}
}

\newcommand{\antlrinl}[1]{\lstinline[language=antlr4,basicstyle=\ttfamily\small]{#1}}

\newcommand{\AppendixSessions}{Appendix~A\xspace}
\newcommand{\AppendixWDN}{Appendix~B\xspace}
\newcommand{\AppendixWDNDomain}{Appendix~C\xspace}

\newcommand{\TheoremPolynomial}{Thm.~A.3 in the Appendix\xspace}
\newcommand{\DefinitionComposition}{Def.~A.1 in the Appendix\xspace}
\newcommand{\TheoremLive}{Thm.~A.2 in the Appendix\xspace}

\title{
iCPS-DL: A Description Language for\\ Autonomic Industrial Cyber-Physical Systems
} 

\author{
    \IEEEauthorblockN{Dimitrios Kouzapas\IEEEauthorrefmark{1},
    Christos Panayiotou\IEEEauthorrefmark{1}\IEEEauthorrefmark{2}},
    Demetrios G. Eliades\IEEEauthorrefmark{1}
    \\
    \IEEEauthorblockA{\IEEEauthorrefmark{1}KIOS Research and Innovation Centre of Excellence, University of Cyprus, Cyprus}
    \\
    \IEEEauthorblockA{\IEEEauthorrefmark{2}Department of Electrical and Computer Engineering, University of Cyprus, Cyprus}

    \thanks{This work has received funding from the European Union’s Horizon Europe research and innovation programme under grant agreement No. 958478 (EnerMan) and supported by the European Union Horizon 2020 program Teaming under Grant Agreement No. 739551 (KIOS CoE) and the Government of the Republic of Cyprus through the Deputy Ministry of Research, Innovation and Digital Policy.}
}

\begin{document}

    \maketitle

    \begin{abstract}
        Modern industrial systems require frequent updates to their cyber and physical infrastructures, often demanding considerable reconfiguration effort. This paper introduces the \textit{industrial Cyber-Physical Systems Description Language}, \iCPSDL, which enables autonomic reconfigurations for industrial Cyber-Physical Systems. The \iCPSDL maps an industrial process using semantics for physical and cyber-physical components, a state estimation model, and agent interactions. A novel aspect is using communication semantics to ensure {\em live} interaction among distributed agents. Reasoning on the semantic description facilitates the configuration of the industrial process control loop. A Water Distribution Networks domain case study demonstrates \iCPSDL's application.
    \end{abstract}

    \begin{IEEEkeywords}
        description language, cyber-physical systems, self-reconfiguration, ontologies, semantics  
    \end{IEEEkeywords}


\section{Introduction}


Industrial Internet of Things (\IIoT) ~\cite{Sisinni20184724} drives Industry 4.0 by optimising key performance indicators through the addition, update, or removal of assets within the industrial ecosystem, often leading to the reconfiguration of the underlying Cyber-Physical System (CPS).
Moreover, as an industrial process scales in terms of size and capabilities, there is an increased possibility for events that disrupt its normal operation to occur (e.g., component failures, or cyber-attacks), requiring reconfiguration efforts~\cite{isermann2006fault}.
The reconfiguration of an industrial system is a knowledge-intensive, time-consuming, and error-prone task. It requires expertise in the industrial process, control system engineering, \IT networking, and understanding of CPS and systems security. A bigger challenge is when this reconfiguration should occur automatically without human intervention while ensuring minimum downtime, maximum productivity, and minimum financial losses.

%
\dkcol{
Distributed industrial automation deploys multiple interacting nodes, enabling the industrial Cyber-Physical Systems (\iCPS) paradigm. 
This paper proposes a framework for the autonomic reconfiguration of CPS.
The framework presents the industrial Cyber-Physical System Description Language (\iCPSDL), which provides a semantic foundation for an autonomic \iCPS architecture, i.e.,~a system designed for continuous self-regulation and self-adaptation~\cite{Kephart2003}.

Specifically, \iCPSDL defines instances of the architecture as ontology metaschemas, called {\em industrial domain definitions}, e.g., water distribution systems, HVAC, and power systems.
An industrial domain defines the physical and cyber-physical components of an industrial system, along with their state estimation capabilities. It also defines a knowledge base of \iCPS agents in terms of their interaction semantics.
A {\em program} in \iCPSDL represents an instance of an industrial domain. The reasoning capabilities of \iCPSDL enable autonomic \iCPS configurations by constructing the state estimation knowledge graph of a program and identifying state estimation redundancies.
A state estimation graph guides agent composition within an \iCPS network, computing controller inputs and dynamically configuring new control loops.

The capabilities of \iCPSDL are demonstrated through a use case from the Water Distribution Network (WDN) domain.
\iCPSDL is released under an Open Source licence\footnote{\url{https://github.com/KIOS-Research/iCPS-DL}}. A CodeOcean\footnote{\url{https://codeocean.com/capsule/1441773/tree/}} 
module includes a proof-of-concept autonomic supervisor controlling a simulation of the paper's examples. 

The rest of the paper is organised as follows:
Sections~\ref{sec:background} and~\ref{sec:related} present background, and related work, respectively. 
Section~\ref{sec:architecture} \dkcol{defines the \iCPSDL framework}.
Section~\ref{sec:sessions} defines the \iCPSDL semantics for communicating agents, whereas
Section~\ref{sec:ontology} defines an ontology meta-schema for defining industrial domains.
Section~\ref{sec:wdn_example} presents the \iCPSDL definition of the WDN domain and applies \iCPSDL reasoning over a corresponding example.
Finally, Section~\ref{sec:conclusion} discusses
future work and concludes.
Theorems and supplementary material are included in the Appendix.
}

\section{Background}
\label{sec:background}
\dkcol{
This section provides a brief background information on the foundations of our work. A detailed background regarding the semantic theory on interactions is provided in Appendix A. 

In this work, state estimation refers to methods for estimating unmeasured system values from available measurements, ranging from classic approaches like Luenberger Observers and Kalman Filters to domain-specific estimators, such as IHISE for hydraulics~\cite{IHISE} and BUBA for water quality~\cite{BUBA}.

Agent composability is based on {\em behavioural types}~\cite{book:behavioural-types}, a type-theoretic framework for agent interactions. Behavioural types structure user-defined interaction while ensuring critical properties such as deadlock-freedom and liveness~\cite{less_is_more}. Their type-theoretic foundation applies broadly to various message-passing programming paradigms.

Technologies such as \IIoT, Industry 4.0, and Analytics have become enablers of CPS implementations~\cite{singh2024cyber}.
Our previous work on the \textit{Semantically-enhanced IoT-enabled Intelligent Control Systems} (SEMIoTICS)~\cite{semiotics_reconfiguration,semiotics_HVAC} reasons over ontological \IIoT descriptions. It semantically describes \iCPS components enabling the composition of feedback control-loop schemes. 
Moreover, \cite{Nicolaou2018,Barrere2020} identify the need to generate alternative control loops by using redundancy in the case of a potential cyber-attack, as well as the creation of alternative configurations to enhance monitoring and control~\cite{Kouzapas2023}.
}


\section{Related work}
\label{sec:related}
Several approaches 
propose architectures for safe, sound, and seamless self-reconfiguration of \iCPS.
For instance, the IEC 61499 standard for distributed automation and control~\cite{iec-61499} defines function blocks as reactive, composable components to build interactive distributed systems. The work in~\cite{SOA-ASM} maps the Service-Oriented Architecture (SOA) to IEC 61499, integrating it with Autonomic Service Management (ASM)~\cite{ibm2005architectural} to propose self-manageable and adaptive \iCPS. The work in~\cite{LYU2024102627} introduces a self-manageable architecture for industrial automation systems by combining multi-agent systems with IEC 61499. The former framework employs a rule-based knowledge base for agent selection and composability, while the latter uses multi-agent modelling.
In contrast, \iCPSDL selects agents based on a state estimation model and composes them according to their interaction semantics.

An ontology is a formal representation of domain knowledge, structured to define concepts and their relationships. Examples include the Open Geospatial Consortium (OGC) Semantic Sensor Networks (SSN) ontology~\cite{haller:hal-02016313} and the European Telecommunications Standards Institute (ETSI) Smart Applications Reference Ontology (SAREF)~\cite{daniele2016interoperability}.
Modelling languages, such as the OGC Sensor Model Language (SensorML)~\cite{sensorml} 
%
and the Systems Modelling Language (SysML)~\cite{delligatti2013sysml} are widely used for \iCPS modelling, whereas the Architecture Analysis and Design Language (AADL)~\cite{AADL} provides semantics for validation and code generation. Modelica~\cite{MATTSSON1998501} models CPS with equations and supports embedded system implementation.
Bridging the two paradigms, knowledge graphs~\cite{knowledge_graph_survey} organise metadata as interconnected nodes, facilitating reasoning and linking semantics to system models.

The need for correct \iCPS configurations drives the development of formal methods. The authors in~\cite{STSL} integrate signal temporal logic with spatial logic to model and validate \iCPS properties. Lingua Franca is a language for deterministic actor interactions used for constructing of verifiable CPS~\cite{Lingua_Franca_CPS}. Event-B offers a set-theoretic framework to verify systems by modelling them as event-reactive state machines~\cite{eventB}. AgentSpeak, a language for modelling multi-agent systems based on Belief-Desire-Intention (BDI) systems~\cite{agentspeak}, is implemented through the Jason framework~\cite{jason}. Furthermore,~\cite{rewrite-cps} introduces a rewriting system to specify and analyse CPS components.
Petri nets are widely used for graphical modelling and analysis of concurrent systems, addressing properties such as state reachability, deadlock freedom, liveness, and fairness. For example,~\cite{Hippo-CPS} employs Petri nets to analyse the control aspects of CPS, while~\cite{petri-net-cps} models CPS as networks of communicating agents using Petri nets.
Contract composition has been applied in the control of dynamic systems, where contracts specify input assumptions and output guarantees, enabling modular and compositional system design~\cite{SHARF2024}.

In comparison, behavioural types are a family of frameworks that offer a comprehensive approach to ensuring the correctness of concurrent interactions while enabling seamless operational integration with \iCPS configuration technologies. Multiparty Session Types~\cite{multiparty_session_types} is a key behavioural types framework for verifying communication properties. Their type-theoretic nature allows integration with several programming languages~\cite{book:behavioural-types} including Java, Haskell, Python, Go, Scala, etc. Other applications include high-performance computing, multiagent systems, code generation, and system monitoring. The work in~\cite{less_is_more} shows that multiparty session-typed agents enjoy, by construction, properties such as liveness and fairness. This is the first work integrating behavioural types in \iCPS.

\section{Autonomic Industrial CPS Architecture}

\label{sec:architecture}


\begin{figure}
    \begin{center}
    \input{tikz/tikz_architecture.tex}
    \end{center}
    \caption{
    The architecture of the autonomic reconfiguration framework. Industrial process, \plant, has a state vector \pstate and inputs actuator signal vector \psignal. A distributed network connects multiple heterogeneous hardware components that perform control, monitoring, and optimisation tasks.
    Hardware components are depicted with white dashed rectangles. Controller, estimators, and sensor agents are depicted as rhombus, square, and circle shapes, respectively.
    \label{fig:architecture}}
\end{figure}


\subsection{Industrial processes and cyber-physical systems}

\dkcol{Fig.~\ref{fig:architecture} depicts the architecture for autonomic industrial processes.}
The industrial process, \plant, is an interconnected network of {\em physical} components, such as a WDN. A {\em state} vector, $\pstate = [\pstate[1], \dots, \pstate[n]]$, characterises the industrial process. Each state $\pstate[i]$, for $1 \leq i \leq n$, quantifies a {\em property} $\property \in \Properties$. States are {\em measured} at specific {\em sensing points} within the industrial process.
A subset of physical components, called {\em actuators}, input a {\em signal} vector, $\psignal = [\psignal[1], \dots, \psignal[m]]$, to {\em control} the state of the industrial process.
A state can be {\em estimated} using an {\em estimator function}, denoted by \efunction, which takes a vector of, measured or estimated, states as input. The collection of estimator functions constitutes the {\em estimation model}, \emodel.

The signal vector is generated by a network of cyber-physical components, depicted by dashed-lined white boxes in Fig.~\ref{fig:architecture}. These components 
include devices such as PLCs, edge computing devices, servers, etc.

Cyber-physical components, due to their computational and networking capabilities, {\em deploy} multiple software {\em agents}, represented in Fig.~\ref{fig:architecture} by circle, square, and rhombus shapes. These agents interact to produce the actuator signal vector \psignal.

Hardware components with sensing capabilities, such as sensor devices, are installed at sensing points to measure physical states. These devices deploy {\em sensor} agents, which measure, digitise, and communicate state information within the network.
Similarly, other CPS devices connect to actuators and deploy {\em controller agents}. Controller agents take an input state vector \estate, which can be either measured or estimated, and generate a signal vector \psignal, sent to the actuators via the underlying hardware.
Additionally, a set of {\em estimator agents} implements the estimation model.

A {\em control loop configuration} deploys distributed interacting agents across the CPS network. The agents produce and communicate a signal vector to the industrial process actuators. These interactions rely on {\em message-passing} communication, implemented using network protocols, such as \MQTT.

\subsection{An Autonomic Supervisor}

The architecture supports an {\em autonomic supervisor} with smart functionalities, such as control loop self-configuration.
The supervisor maintains a {\em knowledge base} with descriptions and rules enabling autonomic capabilities. An {\em industrial domain} is as an ontology schema allowing for semantic descriptions of industrial processes. The industrial domain includes properties, the estimator model, {\em physical component classes}, and an {\em agent repository}. It also links physical components and their interconnections with estimator functions and properties. This association defines a {\em translation function} that constructs a {\em state estimation graph}, which relates states as inputs and outputs of the estimator functions.
The state estimation graph guides the deployment of sensor and estimator agents to compute controller input, and thus configure the control loop.

The autonomic supervisor represents the industrial process as a knowledge graph. This graph captures the physical components, their interconnections, and the hardware components and their capabilities to measure states or control actuators.

A {\em semantic reasoning engine} module analyses an industrial process description to generate a control loop configuration. Given a controller agent with its input state and the actuator to control, the engine supports the following functionalities:
i) translates the industrial knowledge graph into a state estimation graph;
ii) traverses the state estimation graph to identify all {\em estimation trees} containing the information needed to estimate the controller's input state.
iii) identifies the estimator and sensor agents associated with the estimation tree; and
iv) utilises behavioural type theory to compose a deadlock-free and live control loop configuration description.



The {\em control loop configuration deployment} module inputs the semantic description of the control loop configuration produced by the semantic reasoning engine and deploys the corresponding agents. This process involves generating agent code, such as control and estimation logic, along with communication operations. Moreover, it allocates network resources by instructing hardware assets to deploy the generated agents. 

The framework includes an {\em event manager} module that monitors the operation of the industrial process and detects changes. Upon detecting an event, the event manager updates the supervisor knowledge base by performing removal or addition transactions.
If a control system disruption is detected, the event manager invokes the semantic reasoning engine to determine a new control loop configuration.



\subsection{The \iCPSDL: Enabling an Autonomic Supervisor}

The Industrial Cyber-Physical System Description Language (\iCPSDL) was developed using the ANTLR4 parser generator and the Go programming language. 
The language supports user-defined industrial domains, agent repositories, and industrial processes. It also supports the functionalities of the semantic reasoning engine, as shown by the grammar:
\begin{lstlisting}[language=antlr4]
command      : domain | repository | process
             | translate | traverse | configure | ...

domain       : 'domain' '{' domain_decl* '}'
domain_decl  : property | model | class | translation

repository   : 'repository' ID '{' rep_decl+ '}'
rep_decl     : 'estimate'   ID 'using' ID '=' local
             | 'sense'      ID 'using' ID '=' local
             | 'control'    ID 'using' ID '=' local
             | 'actuate'    ID 'using' ID '=' local

process      : PROCESS ID '{' process_decl* '}'
process_decl : device | component | connection_decl
\end{lstlisting}


%
The iCPS-DL GitHub repository 
and CodeOcean module implement a proof of concept autonomic supervisor, with basic implementations for the event manager and the control loop configuration deployment module, that uses the \iCPSDL to control a simulation of the paper's examples. Users can also use the terminal for demonstration purposes, manually defining structures and applying the \iCPSDL functionalities, or to load script files containing \iCPSDL commands.

Currently, the \iCPSDL does not support functionalities for the control loop configuration deployment module. Implementing such functionalities is an important future direction. It could include extending the semantics of \iCPSDL to implement a programming language based on interaction semantics and implementing network communication libraries and tools for remote programming of CPS components.
Moreover, a possible implementation of the event manager could incorporate fault detection algorithms, security incident detection, and manual and automatic component update detection.

\subsection{A Water Distribution Network Example}
\label{subsec:running_example}

Fig.~\ref{fig:motivating_example} presents a simple yet realistic case study of a smart drinking water distribution system serving as a running example throughout the paper.
The system consists of a pump actuator that increases system pressure and a water storage tank that meets the water demands of the area aggregated at the final junction point. These elements are typically geographically distant from each other.
The system includes three cyber-physical devices 
with sensing capabilities to measure and transmit network states, such as pressure and flow. 
The control objective is to maintain the tank water level within user-defined parameters. A water level sensor agent, denoted as \sensor[level], is deployed at the water tank sensing point to measure the water level and transmit it to a controller agent, labelled as \cntrl, deployed to control the pump actuator. The controller 
compares the water level with user-defined upper and lower thresholds, determining whether to send a stop or start signal to the pump actuator.

\begin{figure}
    {
\def\cx{10mm}\def\cy{10mm}\def\myscale{0.8}
\begin{tikzpicture}[x=\cx, y=\cy, scale=0.75, transform shape]
    \tikzstyle{connector}=[rectangle, rotate = 45, draw, fill, minimum height=0.05*\cx, minimum width=0.05*\cy, inner sep=1.5]
        
    \begin{scope}[shift={(0, 0)}, name prefix=cloud-]
        \node[iiot-cloud, minimum width=11*\cx, minimum height=3*\cy]
                                at  (5.5, 2.25) {};

        \node[plc, minimum height=1.75*\cx, minimum width=3*\cy]
                                at  (2, 2.25) (plc1) {};

        \node[plc, minimum height=1.6*\cx, minimum width=3.5*\cy]
                                at  
                                    (5.85, 2.5) (plc2) {};
    
        \node[plc, minimum height=0.75*\cx, minimum width=1.75*\cy]
                                at  (9.5, 2.5) (plc4) {};

        \draw[bus]                  (1, ) -- (10, 1);
        \draw[bus]                  (2, ) -- (plc1);
        \draw[bus]                  (5.85, ) -- (plc2);
        \draw[bus]                  (9.5, ) -- (plc4);

        \node[sensor, flow, inner sep=3]
                                at  (4.5, 2.9) (inflow) {};
        \node[right=0.2*\cx]
                                at  (inflow) {\sensor[inflow]};

        \node[star, star points=6, star point ratio=0.15, fill=black]
                                at   (1.85, 2.5) (shape) {}
                                ;
        \node[below]            at   (shape) {\mysf{init}};
    
        \node[agent, tmass, inner sep=3]
                                at  (3, 2.5) (mass) {}
                                    edge[dotted, thick, pre] (inflow)
                                    edge[dotted, thick, pre] (shape)
                                    ;
        \node[above left=0.1*\cx]
                                at  (mass) {\participant{est}};

        \node[connector]        at  (3, 2) (connector1)    {};
        \draw[dotted, thick]        (mass) -- (connector1);

        \node[controller, inner sep = 3]
                                at  (1, 2) (cntrl) {}
                                ;
        \node                   at  (1.25, 1.75) {\participant{controller}};

        \node[sensor, flow, inner sep=3]
                                at  (6.5, 2) (level) {}
                                    edge[post] (cntrl)
                                    ;

        \node[right=0.2*\cx]
                                at  (level) {\sensor[level]};

        \node[sensor, flow, inner sep=3]
                                at  (9, 2.5) (demand) {}
                                    edge[dotted, thick, post] (mass);

        \node[right=0.2*\cx]
                                at  (demand) {\sensor[demand]};
    \end{scope}

    \begin{scope}[shift={(0, 3.75)}, name prefix=wdn-process-]
        \node[inner sep=0]      at  (0, 1) (j1) {};

        \drawPump{1}{1}{0.3}{pump}{{\pump}}

        \node[cylinder, shape border rotate=90, draw, minimum height=1.2*\cx, minimum width=1*\cy]
                                at  (6.5, 0.75) (tank) {};
        \node                   at  (tank) {\tank};

        \draw[double, ->]           (j1) -- (6, 1);

        \node[jun]              at  (9, 0.5) (j3) {};
        \node[above]            at  (j3) {\mysf{junction}}                        ;

        \draw[double, ->]           (7, 0.5) -- (j3);

        \node[inner sep=0]          at  (6.5, 0.25) (sensor1) {};
        \node[inner sep=0]          at  (4.5, 1)   (sensor3) {};
        \node[inner sep=0]          at  (9, 0.5)  (sensor4){};
    \end{scope}

    \draw[thick, ->]                       (cloud-cntrl) -- (wdn-process-pump);
    \draw[thick, ->]                       (wdn-process-sensor1) -- (cloud-level);
    \draw[thick, ->]                       (wdn-process-sensor3) -- (cloud-inflow);
    \draw[thick, ->]                       (wdn-process-sensor4) -- (cloud-demand);
\end{tikzpicture}
}
    \caption{The drinking water distribution network consists of a pump, a tank, and a junction where water is consumed. Three sensors monitor the tank inflow, tank water level, and actual water demand, along with a controller for the pump actuator. Each sensor agent communicates and exchanges information with other CPS components.
    \label{fig:motivating_example}}
\end{figure}

If a fault occurs in the level sensor 
the pump actuator may malfunction, potentially resulting in an overflow or an empty tank. However, an autonomic supervisor has sufficient information to estimate the tank's water level. For instance, it can deploy an estimator agent that uses the inflow data from sensor agent \sensor[inflow] and outflow data from sensor agent \sensor[demand], along with an initial condition, \mysf{init} (the last known tank level), to estimate the current water level.
A semantic description of the industrial process and its state estimation model guides the autonomic supervisor in inferring the water level estimator and its associated input and output states. Additionally, the communication semantics of each agent enable the autonomic supervisor to configure a live, deployable control loop.


\section{Agent-based Semantic Framework}
\label{sec:sessions}

The semantic framework for agent interactions is based on \textit{behavioural types}~\cite{book:behavioural-types}.
For an introduction to {\em multiparty session types} refer to the tutorial paper~\cite{gentle_multiparty}. \AppendixSessions 
formally introduces the behavioural type theory for \iCPSDL.

\begin{figure*}
\begin{lstlisting}[language=antlr4]
  action : ID '!' ID | ID '?' ID                                    # send receive actions
  local  : 'end'                                                    # inactive
         | action '.' local                                         # sequence
         | action '{' ID ':' local '}' ('or' '{' ID ':' local '}')+ # choice
         | ID '.' local                                             # recursion
         | ID                                                       # label
  local_configuration : 'local' '{' (ID '=' local)+ '}'             # local configuration

  pass   : ID '->' ID ':' ID                                        # message pass  
  global : 'end'                                                    # inactive
         | pass '.' global                                          # message pass
         | pass '{' ID ':' global '}' ('or' '{' ID ':' global '}')+ # choice
         | ID '.' global                                            # recursion
         | ID                                                       # label
  global_configuration : 'global' global                            # global protocol
\end{lstlisting}

\begin{lstlisting}[language=iCPSDL]
  lconfig := local {
    s1          = loop. t.tank_mass!flow. loop
    s2          = loop. t.tank_mass!flow. loop
    t.tank_mass = loop. s1?flow. s2?flow. controller!head. loop
    controller  = loop. t.tank_mass?head. u!signal { ON: loop } or { OFF: loop }
    u           = loop. controller?signal { ON: loop } or { OFF: loop }
  }

  gconfig := global loop. s1->t.tank_mass:flow. s2->t.tank_mass:flow. t.tank_mass->controller: head.
                          controller->u:signal { OFF: loop } or { ON: loop }
\end{lstlisting}
\caption{The \iCPSDL grammar for defining local protocols, local configurations, global protocols, global configurations (above). An \iCPSDL snippet semantically describes a local and a global configuration for a control loop involving a tank estimator (below). \label{fig:session_syntax}}
\end{figure*}

Fig.~\ref{fig:session_syntax} (above) defines the \iCPSDL grammar for behavioural types.
Behavioural types use a textual notation, called {\em local protocol}, to define agent interactions. The syntax for {\em send} and {\em receive} {\em actions} form the core of a local protocol.
A send action has the form \iCPSDLinl{p!type}, representing a communication operation that sends a value of type \iCPSDLinl{type} to agent \iCPSDLinl{p}. Dually, a receive action \iCPSDLinl{p?type} represents a communication operation that receives a value of type \iCPSDLinl{type} from agent \iCPSDLinl{p}.

Local protocol \iCPSDLinl{end} represents the {\em inactive} local protocol, meaning it has no active behaviour.
A {\em sequence} action composition 
indicates an agent performing the communication described by \antlrinl{action}, followed by the behaviour described by local protocol \antlrinl{local}.
A {\em choice} protocol 
defines a choice between multiple local protocols. This choice is made using enumeration labels.
For example, an agent with protocol \iCPSDLinl{p!signal \{ ON: local1 \} or \{ OFF: local2 \}} sends either label \iCPSDLinl{ON} or \iCPSDLinl{OFF} of enumeration \iCPSDLinl{signal} to participant \iCPSDLinl{p}, and then proceeds with either \iCPSDLinl{local1} or \iCPSDLinl{local2}, depending on the label sent.
Conversely, an agent with protocol \iCPSDLinl{p?signal \{ ON: local1 \} or \{ OFF: local2 \}} receives either label \iCPSDLinl{ON} or \iCPSDLinl{OFF} of enumeration \iCPSDLinl{signal} from participant \iCPSDLinl{p}, and then proceeds with either \iCPSDLinl{local1} or \iCPSDLinl{local2}, depending on the label received.
The {\em recursion} local protocol 
defines a loop, where the label \antlrinl{'ID'} serves as the execution jump within the loop.
For example, protocol \iCPSDLinl{loop. p?int. loop} represents a loop where an agent repeatedly receives an integer message, \iCPSDLinl{int}, from participant \iCPSDLinl{p}.

A semantic description for an agent associates the agent's name with a local protocol. For example, the following \iCPSDL code specifies the behaviour of a tank head estimator:
\begin{lstlisting}[language=iCPSDL, basicstyle=\footnotesize\ttfamily]
est = loop. s1?flow. s2?flow. controller!head. loop
\end{lstlisting}
Here, the estimator agent \iCPSDLinl{est} operates in a loop. It first receives a \iCPSDLinl{flow} value from sensor \iCPSDLinl{s1}, followed by another \iCPSDLinl{flow} value from sensor \iCPSDLinl{s2}. It then estimates and sends a \iCPSDLinl{head} value to the agent \iCPSDLinl{controller}.

A set of agents forms a {\em local configuration}, defining the interaction of multiple agents. Fig.~\ref{fig:session_syntax} illustrates local configuration, \iCPSDLinl{lconfig}, describing the running example in Fig.~\ref{fig:motivating_example}, focusing on the case where the control loop deploys a tank level estimator.
%
%
Sensor agents \iCPSDLinl{s1} and \iCPSDLinl{s2} provide inflow and outflow values, respectively, to the \iCPSDLinl{est} estimator agent. The estimator \iCPSDLinl{est} computes and sends the \iCPSDLinl{head} value to the \iCPSDLinl{controller}. The \iCPSDLinl{controller} then transmits a binary signal, \ON or \OFF, to regulate the pump agent \iCPSDLinl{u}.

From a communication perspective, agents within a local configuration synchronise through their dual (send/receive) actions. Communicating agents must adhere to desired properties, such as deadlock freedom and liveness.
Deadlock freedom ensures that a local configuration is either in a state where two agents can synchronise or all agents are inactive.
Liveness, a cornerstone property for distributed systems, guarantees that every non-inactive agent will eventually perform an action, ensuring progress for all allocated components in an interaction.
For a formal definition of communication transition semantics and local configuration properties, see \AppendixSessions. 


A {\em global protocol} provides an alternative perspective on behavioural types, which ensures properties such as deadlock freedom and liveness. Specifically, live local configurations may {\em compose} a global protocol in polynomial time with respect to the local configuration syntactic size
(see \TheoremPolynomial).

Fig.~\ref{fig:session_syntax} provides the \iCPSDL syntax for global protocols.
The {\em message-passing} action, \iCPSDLinl{p->q:type}, is the compositional block for global protocols. The action describes 
agent \iCPSDLinl{p} sending type \iCPSDLinl{type} to agent \iCPSDLinl{q}.
Global protocol \iCPSDLinl{end} has no active behaviour.
A message-passing protocol composes in sequence a message-passing action.
A choice protocol declares a choice between multiple protocols. Concretely, protocol \iCPSDLinl{p->q:signal \{ON: global1\} or \{OFF: global2\}}, describes participant \iCPSDLinl{p} sending label \iCPSDLinl{ON} or \iCPSDLinl{OFF} to participant \iCPSDLinl{q}. Both participants proceed according to the chosen label.
The recursive global protocol defines a loop, where the label \antlrinl{'ID'} serves as the execution jump within the loop.

Global protocol \iCPSDLinl{gconfig} in Fig.~\ref{fig:session_syntax} defines the agent interactions for the running example in Fig.~\ref{fig:motivating_example} which is semantically equivalent to local configuration \iCPSDLinl{lconfig}.

A global protocol describes the behaviour of a local configuration, by {\em projecting} its constituent roles. A projection algorithm results in a live local configuration, and symmetrically, whenever a local configuration is live it may {\em compose} a global protocol 
(cf.~\DefinitionComposition~and~\TheoremLive).

The \iCPSDL provides algorithms for projection and composition. 
Command \iCPSDLinl{project} projects a global protocol, e.g.,
\begin{lstlisting}[language=iCPSDL, basicstyle=\footnotesize\ttfamily]
    project gconfig
\end{lstlisting}
produces local configuration \iCPSDLinl{lconfig}. 
Conversely, command \iCPSDLinl{compose} composes a local configuration, e.g.,
\begin{lstlisting}[language=iCPSDL, basicstyle=\footnotesize\ttfamily]
    compose lconfig
\end{lstlisting}
 will produce global protocol \iCPSDLinl{gconfig}.


\section{A Ontology Meta-schema for Autonomic Industrial Cyber-Physical Systems}
\label{sec:ontology}

This section defines a meta-schema for defining ontology schemas for industrial process domains.
This meta-schema enables users to create domain-specific ontologies for defining knowledge graph representations of industrial processes within these domains.
Formally, an industrial domain 
is defined as:
\[
    \inDomain = \lrangle{\Properties, \emodel, \Classes, \translation, \lrangle{\AG, \modelmap}}
\]
where \Properties is the set of properties, \emodel the set of estimator functions, \Classes is a set of industrial component classes, and \translation is a state estimation translation function. Structure \lrangle{\AG, \modelmap} is an agent repository. Set \AG is the set of agents semantics, and agent mapping \modelmap maps agents to the elements of the industrial domain. An industrial domain provides the information to define industrial processes, \process. These concepts are introduced in detail below.

An industrial component class is defined as:
\[
    \lrangle{\class, \properties, \estimators} \in \Classes,
\]
where \class is the class name, and $\properties \subseteq \Properties$ and $\estimators \subseteq \emodel$ are its attributes. 
Set \properties characterises the properties of each state of the industrial component, while \estimators represents the functions estimating the industrial component states. Classes are partitioned into physical components classes and sensing points classes. Sensing point classes have a single attribute specifying the property measured at that point. Moreover, actuator classes, \ActuatorClasses, are a subset of physical component classes, $\ActuatorClasses \subseteq \Classes$. 

The set of agents, \AG, is partitioned into
estimator agents, \AG[g]; sensor agents, \AG[s]; controller agents, \AG[c]; and actuator agents \AG[a].
The description of agents uses behavioural semantics. 
%
%
%
%
Moreover, agent mapping
\[
    \modelmap: (\AG[g] \to \emodel) \cup (\AG[s] \to \Properties) \cup (\AG[a] \to \ActuatorClasses) \cup (\AG[c] \to \ActuatorClasses),
\]
maps estimator, sensor, and actuator agents to the corresponding estimator functions, properties, and actuator classes, respectively. Also, it maps controller agents to actuator classes, indicating the actuators controlled by each controller.

An industrial process knowledge graph is defined as:
\[
    \process = \lrangle{\GG, \lrangle{\Hardware, \h}}.
\]
Graph $\GG = \lrangle{\VV, \EE}$, called industrial process knowledge graph, consists of a set of vertices $v_1, v_2, \dots \in \VV$ and a set of edges $\EE \subseteq \VV \times \VV$. Set $\VV$ is partitioned into:
i) physical components, $\VV[c]$; and
ii) sensing points, $\Sensors$, where $s, s_1, \dots \in \Sensors$.
Additionally, the set of actuators, $\Actuators$, forms a subset of the physical components, i.e., $\Actuators \subseteq \VV[c]$.
Each physical component is associated with a physical component class, and each sensing point with a sensing point class.

Set \Hardware, with $h_1, h_2 \dots \in \Hardware$, is the set of CPS components. Moreover, relation $\h \subseteq  (\Sensors \cup \Actuators) \times \Hardware$, relates sensing points and actuators with hardware components. For example, it associates a sensing point with a sensor device.

The state estimation translation function, $\translation: \set{\GG: \forall \GG} \to \set{\egraph: \forall \egraph}$ translates an industrial process graph into a state estimation graph, $\egraph = \lrangle{\eV, \eE}$. This graph captures information on measuring or estimating states following the state estimation model.
Set \eV, with $v_1, v_2, \dots \in \eV$, is partitioned into:
the set of state nodes, $\eV[s] = \set{v.\property \setbar v \in \VV, v \text{ instance of } \lrangle{\class, \properties, \estimators}, \property \in \properties}$; the set of estimator nodes, denoted as $\eV[g] = \set{v.\efunction \setbar v \in \VV, v \text{ instance of } \lrangle{\class, \properties, \estimators}, \efunction \in \estimators}$; and the set of sensing points, \Sensors. Set $\eE \subseteq \eV \times \eV$ is the set of edges.

In particular, set \eV[g] corresponds to the estimator attributes of the industrial process graph and set \eV[s], derives from the property attributes, constructing the states of the industrial process graph.
Moreover, set \eE is constructed by:
i) translating each physical component, $v \in \VV$, into a state estimation subgraph;
ii) using each edge $(v_1, v_2) \in \EE$ to interlink these subgraphs, forming a comprehensive state estimation graph.
The state estimation graph links states with estimators establishing state estimation relationships, and links states with sensing points establishing state measurement relationships.
%

\newcolumntype{C}[1]{>{\centering\arraybackslash}m{#1}} 

\begin{figure*}
\begin{tabular}{C{7.4cm}C{9.8cm}}
    \includegraphics[scale=0.45]{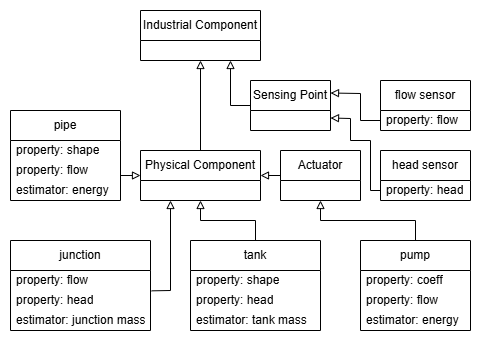}
    &
    \begin{lstlisting}[language=iCPSDL]
agents := repository wdn {
  estimate tank_mass using tmass =
    loop. producer1? flow. producer2? flow. consumer1!head. loop

  estimate junction_mass using jmass =
    loop. producer1? flow. consumer1! flow. loop

  estimate link_energy using lenergy =
    loop. producer1? head. producer2? head. consumer1!flow. loop

  sense head using headSensor = loop. consumer1! head. loop
  sense flow using flowSensor = loop. consumer1! flow. loop

  control pump using controller = loop.
    producer1? head. consumer1!signal {ON: loop} or {OFF: loop}

  actuate pump using pumpActuator =
    loop. producer1? signal {ON: loop} or {OFF: loop}
}
    \end{lstlisting}
    \end{tabular}
    \caption{Class diagram of the Water Distribution Network domain together with \iCPSDL description of the agent repository. \label{fig:WDN_domain}}
\end{figure*}

\section{An Ontology Schema for the WDN domain} 
\label{sec:wdn_example}

This section introduces an \iCPSDL ontology schema for the WDN domain.
Moreover, it introduces the semantics for describing WDN industrial processes through the running example 
For a formal description refer to \AppendixWDN. 

\subsection{The WDN industrial domain}

The following code defines a part of the WDN domain:

\begin{lstlisting}[language=iCPSDL]
wdn := domain {
  property flow, head, tank_shape, link_shape,
           signal {ON, OFF}

  model tank_mass, junction_mass, demand_mass, link_energy

  physical junction(head, flow, junction_mass):
    flow -> junction_mass, junction_mass -> flow

  physical demand(head, flow, demand_mass):
    flow -> demand_mass, demand_mass -> flow

  physical pipe(link_shape, flow, link_energy):
    link_shape -> link_energy, link_energy -> flow

  physical tank (tank_shape, head, tank_mass):
    tank_shape -> tank_mass, tank_mass -> head

  actuator pump(link_shape, flow, link_energy):
    link_shape -> link_energy, link_energy -> flow

  translation pipe -> junction :
    pipe.flow -> junction.junction_mass,
    junction.junction_mass -> pipe.flow,
    junction.head->pipe.link_energy

  translation pipe -> tank :
    pipe.flow -> tank.tank_mass,
    tank.head -> pipe.link_energy

  translation pump -> junction :
    pump.flow -> junction.junction_mass,
    junction.junction_mass -> pump.flow,
    junction.head->pump.link_energy
  ...
}
\end{lstlisting}

Properties include \iCPSDLinl{flow}, \iCPSDLinl{head}, \iCPSDLinl{tank\_shape}, and \iCPSDLinl{link\_shape}, as well as an enumeration, \iCPSDLinl{signal} with \iCPSDLinl{ON} and \iCPSDLinl{OFF} labels, for controlling the pump actuator. The estimator model includes the \iCPSDLinl{tank\_mass}, \iCPSDLinl{junction\_mass}, \iCPSDLinl{demand\_mass}, and \iCPSDLinline{link\_energy}.
There are four physical component classes: \iCPSDLinl{junction}, \iCPSDLinl{demand}, \iCPSDLinl{pipe}, and \iCPSDLinl{tank} including their property and estimator attributes. 
It also defines actuator class \iCPSDLinl{pump}. Each class defines a part of the translation function, establishing state and estimator relations. For example, class \iCPSDLinl{pipe}, defines \iCPSDLinl{link\_shape} state as an input to the estimator \iCPSDLinl{link\_energy} and state \iCPSDLinl{flow} as an output.
The rest of the domain definitions are rules defining the translation function. Each rule takes a connection between two classes and defines corresponding connections between states and estimators, creating the state estimation graph of an industrial process. 
Fig.~\ref{fig:WDN_domain} defines a class diagram and an agent repository for the \iCPSDLinl{wdn} domain.
It defines and maps agents for the \iCPSDLinl{tank\_mass}, \iCPSDLinl{junction\_mass}, and \iCPSDLinl{link\_energy} model estimators, as well as agents for the \iCPSDLinline{head} and \iCPSDLinline{flow} sensors. A controller and a pump actuator agent are also defined. \AppendixWDNDomain 
gives the complete WDN domain definition.

\subsection{Autonomic supervision for the running example}

This section demonstrates the autonomic reconfiguration the running example in response to events.
Additionally, the CodeOcean repository provides a proof-of-concept autonomic supervisor that reconfigures a simulation of the running example, ensuring control despite consecutive failures.

\begin{figure}
    \begin{tabular}{c}
        \begin{tikzpicture}[x=10mm, y=10mm]
    \begin{scope}[scale=0.6]
        \node[junction]     at  (0, 2.5) (r) {};
        \node[above=0.5mm]  at  (r) {\iCPSDLinl{r}};
    
        \node[pump]         at  (2, 2.5) (pmp) {}
                                edge[pre] (r);
        \node[above=0.5mm]  at  (pmp) {\iCPSDLinl{u}};
    
        \node[junction]     at  (4, 2.5) (j) {}
                                edge[pre] (pmp);
        \node[above=0.5mm]  at  (j) {\iCPSDLinl{j}};
    
        \node[pipe]         at  (6, 2.5) (p1) {}
                                edge[pre] (j);
        \node[above=0.5mm]  at  (p1) {\iCPSDLinl{p1}};
    
        \node[tank]         at  (8, 2.5) (tank) {}
                                edge[pre] (p1);
        \node[above=0.5mm]  at  (tank) {\iCPSDLinl{t}};
    
        \node[pipe]         at  (10, 2.5) (p2) {}
                                edge[pre] (tank);
        \node[above=0.5mm]  at  (p2) {\iCPSDLinl{p2}};
    
        \node[junction]     at  (12, 2.5) (d) {}
                                edge[pre] (p2);
        \node[above=0.5mm]  at  (d) {\iCPSDLinl{d}};
                                    
    
        \node[spoint, head]
                                at  (0, 1.5) (s1) {}
                                    edge[pre] (r);
        \node[right=0.5mm]      at  (s1) {\iCPSDLinl{s1}};
    
        \node[spoint, flow]
                                at  (2, 1.5) (s2) {}
                                    edge[pre] (pmp);
        \node[right=0.5mm]      at  (s2) {\iCPSDLinl{s2}};
    
        \node[spoint, head]
                                at  (4, 1.5) (s3) {}
                                    edge[pre] (j);
        \node[left=0.5mm]       at  (s3) {\iCPSDLinl{s3}};
    
        \node[spoint, flow]
                                at  (4, 0.5) (s4) {}
                                    edge[pre, bend right] (j);
        \node[right=0.5mm]      at  (s4) {\iCPSDLinl{s4}};
        
        \node[spoint, flow]
                                at  (6, 1.5) (s5) {}
                                    edge[pre] (p1);
        \node[right=0.5mm]      at  (s5) {\iCPSDLinl{s5}};
    
        \node[spoint, head]
                                at  (8, 1.5) (s6) {}
                                    edge[pre] (tank);
        \node[right=0.5mm]      at  (s6) {\iCPSDLinl{s6}};
    
        \node[spoint, flow]
                                at  (10, 1.5) (s7) {}
                                    edge[pre] (p2);
        \node[right=0.5mm]      at  (s7) {\iCPSDLinl{s7}};
    
        \node[spoint, flow]
                                at  (12, 1.5) (s8) {}
                                    edge[pre] (d);
        \node[right=0.5mm]      at  (s8) {\iCPSDLinl{s8}};
    
        \node[plc,minimum width=0.8cm, minimum height=0.8cm]  at  (2, 0) (plc1) {}
                                    edge[pre]   (s1)
                                    edge[pre]   (s2)
                                    edge[pre]   (s3)
                                    edge[pre]   (s4)
                                    edge[pre, bend left]  (pmp);
        \node                 at    (plc1) {\iCPSDLinl{dev1}};
        
        \node[plc,minimum width=0.8cm, minimum height=0.8cm]  at  (8, 0) (plc2) {}
                                    edge[pre] (s5)
                                    edge[pre] (s6)
                                    edge[pre] (s7);
        \node                 at    (plc2) {\iCPSDLinl{dev2}};
    
        \node[plc,minimum width=0.8cm, minimum height=0.8cm]  at  (12, 0) (plc3) {}
                                    edge[pre] (s8);
        \node                 at    (plc3) {\iCPSDLinl{dev3}};

        \draw[dotted, <->]          (plc1) edge (plc2);
        \draw[dotted, <->]          (plc2) edge (plc3);
    \end{scope}
\end{tikzpicture}
        \\
        \begin{tikzpicture}[x=10mm, y=10mm]
    \begin{scope}[scale=0.6]
        \junctionGraph{-0.5}{0}{j1}{}
        \pumpGraph{1.5}{0}{pump}{}
    
        \draw           (agpump) edge[pre] (prj1)
                        (flj1) edge[<->] (flpump);
    
        \junctionGraph{3.5}{0}{j2}{}
    
        \draw           (agpump) edge[pre] (prj2)
                        (flpump) edge[<->] (flj2);
    
        \pipeGraph{5.5}{0}{p1}{} 
    
        \draw           (flj2) edge[<->] (flp1)
                        (prj2) edge[post] (agp1);
    
        \tankGraph{7.5}{0}{tank}{} 
    
        \draw           (agtank) edge[<->] (flp1)
                        (prtank) edge[post] (agp1);
        
        \pipeGraph{9.5}{0}{p2}{} 
        \draw           (agtank) edge[<->] (flp2)
                        (prtank) edge[post] (agp2);
        
        \junctionGraph{11.5}{0}{j3}{} 
        \draw           (flp2) edge[<->] (flj3)
                        (agp2) edge[pre] (prj3);
                        
    
        \node[spoint, head]
                        at  (0, 3) (s1) {}
                            edge[post, bend right] (prj1);
        \node[right]    at  (s1) {\iCPSDLinl{s1}};
    
        \node[spoint, flow]
                        at  (2, 3) (s2) {}
                            edge[post] (flpump);
        \node[right]    at  (s2) {\iCPSDLinl{s2}};
    
        \node[spoint, head]
                        at  (4, 3) (s3) {}
                            edge[post, bend right] (prj2);
        \node[left]     at  (s3) {\iCPSDLinl{s3}};
    
        \node[spoint, flow]
                        at  (4.75, 2.5) (s4) {}
                            edge[post, bend left] (demj2);
        \node[right]    at  (s4) {\iCPSDLinl{s4}};
        
        \node[spoint, flow]
                        at  (6, 3) (s5) {}
                            edge[post] (flp1);
        \node[right]    at  (s5) {\iCPSDLinl{s5}};
    
        \node[spoint, head]
                        at  (8, 3) (s6) {}
                            edge[post, bend left] (prtank);
        \node[right]    at  (s6) {\iCPSDLinl{s6}};
    
        \node[spoint, flow]
                        at  (10, 3) (s7) {}
                            edge[post] (flp2);
        \node[left]     at  (s7) {\iCPSDLinl{s7}};
    
        \node[spoint, flow]
                        at  (12, 3) (s8) {}
                            edge[post, bend left] (demj3);
        \node[right]    at  (s8) {\iCPSDLinl{s8}};
    \end{scope}
\end{tikzpicture}
        \\
        \begin{tikzpicture}[x=10mm, y=10mm]
    \begin{scope}[scale=0.6]
            \node[junction]     at  (0, 8.25) (r) {}
                                ;
            \node[above=0.5mm]  at  (r) {\iCPSDLinl{r}};

            \node[pump]         at  (2, 8.25) (pump) {}
                                ;
            \node[above=0.5mm]  at  (pump) {\iCPSDLinl{u}}
                                ;                
            \draw               (r) edge[post] (pump);
        
            \node[junction]     at  (4, 8.25) (j) {}
                                    edge[pre] (pump);
            \node[above=0.5mm]  at  (j) {\iCPSDLinl{j}};
        
        
            \node[pipe]         at  (6, 8.25) (p1) {}
                                    edge[pre] (j);
            \node[above=0.5mm]  at  (p1) {\iCPSDLinl{p1}};
        
            \node[tank]
                                at  (8, 8.25) (tank) {}
                                    edge[pre] (p1)
                                ;            
            \node[above=0.5mm]  at  (8, 8.25) {\iCPSDLinl{t}};

            \node[pipe]         at  (10, 8.25) (p2) {}
                                    edge[pre] (tank)
                                ;
            \node[above=0.5mm]  at  (p2) {\iCPSDLinl{p2}};
        

            \node[junction]     at  (12, 8.25) (d) {}
                                    edge[pre] (p2)
                               ;

            \node[above=0.5mm]  at  (d) {\iCPSDLinl{d}};

        \node[spoint, head]
                            at      (0, 7.25) (sp1) {}
                                    edge[pre, thick] (r);
        \node[right]        at      (sp1) {\iCPSDLinl{s1}};
        
        \node[spoint, flow]
                            at      (2.5, 7.25) (sp2) {}
                                    edge[pre, thick] (pump);
        \node[left]         at      (sp2) {\iCPSDLinl{s2}};
        
        \node[spoint, head]
                            at      (3.5, 7.25) (sp3) {}
                                    edge[pre, thick] (j);
        \node[right]        at      (sp3) {\iCPSDLinl{s3}};
        
        \node[spoint, flow]
                           at      (4.5, 7.25) (sp4) {}
                                    edge[pre, thick] (j);
        \node[right]        at      (sp4) {\iCPSDLinl{s4}};
        
        \node[spoint, flow]
                            at      (6, 7.25) (sp5) {}
                                    edge[pre, thick] (p1);
        \node[right]        at      (sp5) {\iCPSDLinl{s5}};
        
        \node[spoint, head] at      (8, 7.25) (sp6) {}
                                    edge[pre, thick] (tank);
        \node[right]        at      (sp6) {\iCPSDLinl{s6}};
        
        \node[spoint, flow]
                            at      (10, 7.25) (sp7) {}
                                    edge[pre, thick] (p2);
        \node[right]        at      (sp7) {\iCPSDLinl{s7}};
        
        \node[spoint, flow]
                            at      (12, 7.25) (sp8) {}
                                    edge[pre, thick] (d);
        \node[right]        at      (sp8) {\iCPSDLinl{s8}};

        \draw[->, thick]            (1.25, 7.25) -- (pump);

        \node[iiot-cloud, minimum width=140mm, minimum height=70mm, scale=0.6]   at  (6.25, 3) {};
        
        \draw[plc]              (-0.5, 0) rectangle (5, 6);
        \node[below]        at  (2.5, 0) {\iCPSDLinl{dev1}};
        
        \draw[plc]              (5.5, 0) rectangle (10.5, 6);
        \node[below]        at  (8, 0) {\iCPSDLinl{dev2}};
        
        \draw[plc]              (11, 3) rectangle (13, 6);
        \node[below]        at  (12, 3) {\iCPSDLinl{dev3}};

        \node[sensor, head]
                            at      (0, 5.5) (s1) {}
                                    edge[pre, thick] (sp1);
        
        \node[sensor, flow]
                            at      (2.5, 5.5) (s2) {}
                                    edge[pre, thick] (sp2);
        
        \node[sensor, head]
                            at      (3.5, 5.5) (s3) {}
                                    edge[pre, thick] (sp3);
        
        \node[sensor, flow]
                           at      (4.5, 5.5) (s4) {}
                                    edge[pre, thick] (sp4);
        
        \node[sensor, flow]
                            at      (6, 5.5) (s5) {}
                                    edge[pre, thick] (sp5);
        
        \node[sensor, head] at      (8, 5.5) (s6) {}
                                    edge[pre, thick] (sp6);
        
        \node[sensor, flow]
                            at      (10, 5.5) (s7) {}
                                    edge[pre, thick] (sp7);
        
        \node[sensor, flow]
                            at      (12, 5.5) (s8) {}
                                    edge[pre, thick] (sp8);
        
        
        
        \node[head, pstate] at  (0, 3.5) (pr1) {}
                                edge[pre, darkyellow, thick] (s1);
        \node[left]         at  (0, 4.75) {$y_{1}$};
        
        \node[flow, pstate]
                            at  (2.5, 1.5) (pumpFlow) {}
                                edge[pre, darkgreen, thick] (s2);
        \node[left]         at  (2.5, 4.75) {$y_{2}$};
        
        \node[head, pstate] at  (3.5, 3.5) (pr2) {}
                                edge[pre, darkyellow, thick] (s3);
        \node[left]         at  (3.5, 4.75) {$y_{3}$};
        
        \node[flow, pstate]
                            at  (4.5, 3.5) (pumpDemand) {}
                                edge[pre, darkgreen, thick] (s4);
        \node[left]         at  (4.5, 4.75) {$y_{4}$};

        \node[jmass, estimator]
                            at  (4.5, 1.5) (pumpjMass){}
                                edge[pre, darkgreen, thick] (pumpDemand)
                                edge[pre, thick] node[above, black] {$\hat{y}_2$} (pumpFlow)
                            ; 
        
        \node[energy, estimator]
                            at  (2, 3.5) (pumpEnergy) {}
                                edge[pre, darkyellow, thick] node[above, black, pos=0.65] {$y_1$} (pr1)
                                edge[post, darkyellow, thick] node[left, black] {$\hat{y_2}$} (pumpFlow)
                                edge[pre, darkyellow, thick] node[above, black, pos=0.75] {$y_3$} (pr2)
                            ;
        
        \node[controller]   at  (1.25, 0.5) (controller) {};
        \draw[thick]            (controller) -- (1.25, 7.25); 
        
        \node[left]         at  (1.25, 4.75) {$v$};
        
        
        \node[flow, pstate]
                            at  (6, 1.5) (tankIflow) {}
                                edge[pre, red, thick] (s5);
        \node[left]         at  (6, 4.75) {$y_{5}$};
        
        \node[head, pstate]
                            at  (8, 0.5) (tankHead) {}
                                edge[pre, blue, thick] (s6);
        \node[left]         at  (8, 4.75) {$y_{6}$};
        
        \node[flow, pstate]
                            at  (10, 3.5) (tankOflow) {}
                                edge[pre, red, thick] (s7);
        \node[left]         at  (10, 4.75) {$y_{7}$};
        
        \node[tmass, estimator]
                            at  (10, 1.5) (tankMass) {}
                                edge[pre, red, thick] node[black, left] {$y_{7}$} (tankOflow)
                                edge[pre, red, thick] node[above, black, pos=0.75] {$\hat{y}_{5}$} (tankIflow);
        
        \path[draw, red, thick, ->]  (tankMass)  -- node[black, left] {$\hat{y}_{6}$}
                                (10, 0.5) -- (tankHead);
        
        \node[flow, pstate]
                            at  (12, 3.5) (demand) {}
                                edge[pre, darkgreen, thick] (s8);
        \node[left]         at  (12, 4.75) {$y_{8}$};
        
        
        \draw                   (demand) edge[post, darkgreen, thick] node[above] {$y_{8}$} (tankOflow);
        \draw                   (pumpjMass) edge[post, thick] node[above, black] {$\hat{y}_5$} (tankIflow);
        
        \draw                   (tankHead) edge[post, thick] node[above, pos=0.4] {$\hat{y}_{6}$} (controller);
        
    \end{scope}
\end{tikzpicture}
    \end{tabular}
    \caption{
        Top: Graphical representation 
        of the running example in Section~\ref{subsec:running_example}.
        Middle: Graphical representation of the
        state estimation graph.
        Bottom: Graphical representation of
        the seven estimation trees 
        rooted at $\tank.\head$,
        with a hardware device assignment for each node. 
 	\label{fig:WDN_IP}
    }
\end{figure}

Fig.~\ref{fig:WDN_IP} (top)  depicts the knowledge graph of the running example in Section~\ref{subsec:running_example} extended with additional sensing points. 
The corresponding \iCPSDL description is defined as:
\begin{lstlisting}[language=iCPSDL, caption={\iCPSDL description of the running example. Shaded elements are affected in case of device \iCPSDLinl{dev2} failure.}, label={lst:iCPSDL_process}]
simple := process wdn {
  device   dev1,(*\grayshade{dev2,}*) dev3

  physical r, d        demand
  physical j           junction
  physical p1, p2      pipe
  physical t           tank
  actuator u@dev1      pump

  sensor  s1@dev1, s3@dev1,(*\grayshade{s6@dev2}*) head
  sensor  s2@dev1, s4@dev1,(*\grayshade{s5@dev2, s7@dev2,}*) s8@dev3 flow

  conn j1->u, u->j2, j2->p1, p1->t, t->p2, p2->j3, j1->s1,
       u->s2, j2->s3, j2->s4,(*\grayshade{p1->s5, t->s6, p2->s7,}*)j3->s8
}
\end{lstlisting}
The process specifies the \iCPSDLinl{wdn} industrial domain. It includes three hardware devices \iCPSDLinl{dev1}, \iCPSDLinl{dev2}, and \iCPSDLinl{dev3}.
The system features two demand points \iCPSDLinl{r} and \iCPSDLinl{d}, a junction \iCPSDLinl{j}, two pipes \iCPSDLinl{p1} and \iCPSDLinl{p2}, a tank \iCPSDLinl{t}, and a pump actuator \iCPSDLinl{u} controlled by device \iCPSDLinl{dev1}.
Additionally, three head sensors and five flow sensors are deployed at their corresponding sensing devices. The process also defines the interconnections between components.

The following \iCPSDL code:
\begin{lstlisting}[language=iCPSDL, caption={Semantic Reasoning using iCPS-DL}, label={lst:iCPSDL_reasoning}]
    seg := translate simple
    trees := traverse t.head seg
    lconfig := configure trees[1] agents controller u
    gconfig := compose lconfig
\end{lstlisting}
translates process \iCPSDLinl{simple} into the corresponding state estimation graph, \iCPSDLinl{seq}, depicted in Fig.~\ref{fig:WDN_IP}~(middle). Each physical component is translated into the state (pentagon shapes) and estimator nodes (square shapes) as defined by its class. The translation function then interconnects state nodes, estimator nodes, and sensing points (triangle shapes). It then identifies a forest of estimation trees, rooted at node \iCPSDLinl{t.head}, by traversing state estimation graph \iCPSDLinl{seg}.
Fig.~\ref{fig:WDN_IP} (bottom) presents an overlay of the seven identified estimation trees. Each tree details an estimation or measurement of the \iCPSDLinl{t.head} property. Blue and red highlight the estimation trees corresponding to the two control schemes described in Section~\ref{subsec:running_example}. The diagram also maps tree nodes to hardware devices.
Assuming that \iCPSDLinl{trees[1]} accesses the red estimation tree, the next command uses the agent repository \iCPSDLinl{agents} to produce local configuration \iCPSDLinl{lconfig} from Fig.~\ref{fig:session_syntax}.
The last command composes global protocol \iCPSDLinl{gconfig} validating that \iCPSDLinl{lconfig} is live. The control loop configuration deployment module will generate and deploy within the \iCPS network the agents corresponding to \iCPSDLinl{lconfig}. 

Assume now that device \iCPSDLinl{dev2} presents a failure. The event manager will detect the failure and update the description of process \iCPSDLinl{simple} without device \iCPSDLinl{dev2} and sensors, \iCPSDLinl{s5}, \iCPSDLinl{s6}, and \iCPSDLinl{s7}, i.e., the shaded elements in Listing~\ref{lst:iCPSDL_process}. It will then invoke the semantic reasoning engine to run code in Listing~\ref{lst:iCPSDL_reasoning} and produce a new control loop configuration. Specifically, it identifies two estimation trees that utilise the tank mass estimator with sensor \iCPSDLinl{s8} measuring the tank outflow. Additionally, the tank inflow is estimated using the junction mass estimator for \iCPSDLinl{j}, with demand measured at sensor \iCPSDLinl{s4}. One tree measures the junction mass inflow at sensor \iCPSDLinl{s2}, while the other estimates it using the link estimator for the pump, based on the head measurements from sensors \iCPSDLinl{s1} and \iCPSDLinl{s3}.



\section{Conclusion and Future Work}
\label{sec:conclusion}
\dkcol{
This work introduced a framework for the autonomic reconfiguration of CPS. It presents \iCPSDL, a language that enables an autonomic supervisor to describe and reason over \iCPS control loop configurations. The \iCPSDL provides semantics for describing industrial domain ontology schemas, which are used to define industrial process knowledge graphs within the domain. The semantics also define agent interactions using behavioural types. Reasoning over the knowledge graph can identify a set of agents, whose deployment can configure a control loop. Moreover, behavioural types theory ensures safe and live agent interaction.
The \iCPSDL expressive capabilities are demonstrated through a representation of the WDN domain while its autonomic enabling capabilities are showcased via an instructive example from this domain.

Future work focuses on conducting in-depth evaluation of the \iCPSDL framework using the KIOS Water Network Testbed \cite{VRACHIMIS2022655}, to validate its flexibility and robustness. 
Additionally, the type-theoretic foundations of \iCPSDL will serve as a basis for a new toolchain~\cite{book:behavioural-types} that implements the modules of the autonomic supervisor. 
This toolchain may include: 
code generation by translating control loop local configurations into agent code templates that incorporate communication operations and algorithmic logic;
communication libraries facilitating the composability and deployment of agents;
a CPS programming language integrating behaviourally typed communication primitives and knowledge graph reasoning; and
a tool offering visual representations of behavioural types to provide insights into control loop configurations.
}

    \bibliographystyle{ieeetr}
    \bibliography{references}           

\appendices

\newcommand{\figMotivatingExample}{Fig.~2 of the main paper\xspace}
\newcommand{\secMotivatingExample}{Sec.~II-D of the main paper\xspace}

\section{A Semantic Theory for Communication Interactions}

\label{app:sec:sessions}

This section presents the theoretical framework for the agent interaction semantics of \iCPSDL.
The semantic framework for agent interactions is based on
\textit{behavioural types}~\cite{book:behavioural-types},
which is a family of frameworks for
semantic reasoning over message passing communication. 

Message passing can be represented
using state machines. 
State machine actions describe the ability of an
agent to send messages to, or receive messages
from its communicating adversaries. 
Instead of the usual quintuple representation of
state machines,
behavioural types use a textual
notation, called
{\em local protocol} 
to define agent interaction. For an introduction to {\em multiparty session types}, as used in this section, the reader is referred to the following tutorial paper~\cite{gentle_multiparty}.

\begin{table*}
{\small
    \begin{tabular}{ll}
        \hline
        \hline
        participants & $\Participants = \set{\p. \q, \dots}$
        \tstrut
        \bstrut
        \\
        \hline
        message types & $\U \in \Types = \set{\nat, \bool, \ON, \OFF, \dots}$
        \tstrut
        \bstrut
        \\
        \hline
        alphabet & $\element \in \Participants \times \set{\outputSymbol, \inputSymbol} \times \Types = \alphabet$
        \tstrut
        \bstrut
        \\
        \hline
        local types
            &  $\local \in \Local = \set{\tinact}
            \cup
            \set{\element \sep \local \setbar \element \in \alphabet, \local \in \Local}
            \cup
            \set{\local[1] \tor \local[2] \setbar \local[1], \local[2] \in \Local}
            \cup
            \set{\trec{t} \local \setbar \local \in \Local}
            \cup \set{\tvar{t}, \tvar{t_1}, \dots}$
        \tstrut
        \bstrut
        \\
        \hline
        participants
        & $\participants: \Local \times 2^{|\Participants|}$
        \tstrut
        \bstrut
        \\ function &
        $
        \participants[\tinact] = \es, \qquad 
        \participants[\tvar{t}] = \es, \qquad
        \participants[\out{\p}{\U} \local] = \set{\p} \cup \participants[\local], \qquad
        \participants[\inp{\p}{\U} \local] = \set{\p} \cup \participants[\local],$
        \\
        & $\participants[{\local[1] \tor \local[2]}] = \participants[{\local[1]}] \cup \participants[{\local[2]}], \qquad
        \participants[\trec{t} \local] = \participants[\local]$
        \bstrut
        \\
        \hline
        syntactic size & $\size: \Local \to \mathbb{N}$
        \tstrut
        \bstrut
        \\
        &
        $\size[\tinact] = \size[\tvar{t}] = 1,
        \qquad
        \size[\out{\p}{\U} \local] = \size[\inp{\p}{\U} \local] = 1 + \size[\local],
        \qquad
        \size[{\local[1] \tor \local[2]}] = \size[{\local[1]}] + \size[{\local[2]}],
        \qquad
        \size[\trec{t} \local] = \size[\local]$
        \bstrut
        \\
        \hline
        substitution &
        $[\substff]: \Local \times \Local \times \set{\tvar{t}, \tvar{t_1}, \dots} \to \Local$
        \tstrut
        \bstrut
        \\ &
        $\tinact \substt{\local}{\tvar{t}} = \tinact,$ \qquad
        $\tvar{t} \substt{\local}{\tvar{t}} = \local$, \qquad
        $\tvar{t'} \substt{\local}{\tvar{t}} = \tvar{t'}$ if $\tvar{t} \not= \tvar{t'}$, \qquad
        $\element \sep \locald \substt{\local}{\tvar{t}} = \element \sep (\locald \substt{\local}{\tvar{t}})$
        \\[1mm]
        & $(\local[1] \tor \local[2]) \substt{\local}{\tvar{t}} = (\local[1] \substt{\local}{\tvar{t}}) \tor (\local[2] \substt{\local}{\tvar{t}})$,
        \qquad
        $\trec{t'} \locald \substt{\local}{\tvar{t}} = 
        \left\{
        \begin{array}{ll}
            \trec{t'} (\locald \substt{\local}{\tvar{t}}) & \text{ if } \tvar{t} \not= \tvar{t'}
            \\
            \trec{t'} \locald & \text{ if } \tvar{t} = \tvar{t'}
        \end{array}
        \right.$
        \\[3mm]
        \hline
        transition & $\by \subseteq \Local \times \alphabet \times \Local$
        \tstrut
        \bstrut
        \\ &
            $\element \sep \local \by[\element] \local$, for all $\element \sep \local \in \Local$,
        \qquad
            $\local[1] \tor \local[2] \by[\element] \locald$ if $\local[1] \by[\element] \locald$ or $\local[2] \by[\element] \locald$,
        \\ &
            $\trec{t} \local \by[\element] \locald$ if $\local \substt{\trec{t} \local}{\tvar{t}} \by[\element] \locald$
        \bstrut
        \\
        \hline
        alphabet &
            $\act{\p}{\q}{\U} \in \calphabet = \Participants \times \set{\passSymbol} \times \Participants \times \Types$
        \tstrut
        \bstrut
        \\
        \hline
        local configuration
        & $\LL \in \LSet = \set{\LL \setbar \LL: \Participants \rightharpoonup \Local}$
        \tstrut
        \bstrut
        \\
        \hline
        active participants &
            $\activef: \LSet \to 2^{|\Participants|}$
        \tstrut
        \bstrut
        \\ &
            $\activep{\LL} = \domain{\LL} \backslash \set{\p \setbar \p: \tinact \in \LL}$
        \bstrut
        \\
        \hline
        syntactic size & $\size: \LSet \to \mathbb{N}$
        \tstrut
        \bstrut
        \\
        & $\size[\set{\p[i]: \local[i] \setbar i \in I}] = \sum\limits_{i \in I} \size[{\local[i]}]$
        \bstrut
        \\
        \hline
        communication &
            $\by \subseteq \LSet \times \calphabet \times \LSet$
        \tstrut
        \bstrut
        \\ transition &
            $\LL, \p: \local[1], \q: \local[2] \by[\act{\p}{\q}{\U}] \LL, \p: \locald[1], \q: \locald[2]$
            whenever
            $\local[1] \by[\actOut{\q}{\U}] \locald[1]$
            and
            $\local[2] \by[\actInp{\p}{\U}] \locald[2]$
        \bstrut
        \\
        \hline
        deadlock-freedom & \df{\LL}
        \tstrut
        \bstrut
        \\
        &   i) $\exists \LLd, \element \in \calphabet: (\LL \by[\element] \LLd)$ or $\activep{\LL} = \es$; and
        \\
        &   ii) $\forall \LLd: [(\exists \element \in \calphabet, \LL \by[\element] \LLd) \implies \df{\LLd}]$
        \\
        liveness & \live{\LL}
        \tstrut
        \bstrut
        \\ &
            i) $\forall \p \in \activep{\LL}, \exists n \geq 0: [\exists \LL[1], \dots, \LL[n], \LLd$ and 
            $\exists \element[1], \dots, \element[n], \element \in \calphabet:$
        \\ &
            $(\LL \by[{\element[1]}] \LL[1], \dots, \LL \by[{\element[n]}] \LL[n]$ and
            $\LL[n] \by[\act{\p}{\q}{\U}] \LLd or \LL[n] \by[\act{\q}{\p}{\U}] \LLd)]$; and
        \\
            &
                ii) $\forall \LLd, [(\exists \element \in \calphabet, \LL \by[\element] \LLd) \implies \live{\LLd}]$
        \bstrut
        \\
        \hline
        global &
            $\glob \in \Glob = \set{\ginact}
        \cup
            \set{\element \sep \glob \setbar \element \in \calphabet, \glob \in \Glob}
        \cup
            \set{\glob[1] \gor \glob[2] \setbar \glob[1], \glob[2] \in \Glob}
        \cup
            \set{\rec{t} \glob \setbar \glob \in \Glob} \cup \set{\var{t}, \var{t_1}, \dots}$
        \tstrut
        \bstrut
        \\
        \hline
        projection/ &
        $\proves \subseteq \LSet \times \Glob$
        \tstrut
        \\[2mm]
        composition & \multicolumn{1}{c}{
            $\set{\p[i]: \tinact \setbar i \in I} \proves \ginact$
            \qquad \quad
            $\set{\p[i]: \tvar{t} \setbar i \in I} \proves \var{t}$
        }
        \\[2mm]
        & \multicolumn{1}{c}{
            $\tree{
                \LL, \p: \local, \q: \locald \proves \glob
            }{
                \LL, \p: \out{\q}{\U} \local, \q: \inp{\p}{\U} \locald \proves \pass{\p}{\q}{\U} \glob
            }$
        \qquad
            $\tree{
                \forall \in I,\ \LL, \p: \local[i], \q: \locald[i] \proves \globd[i]
            }{
                \LL, \p: \Tor{i \in I} \out{\q}{\U[i]} \local[i], \q: \Tor{i \in I} \inp{\p}{\U[i]} \locald[i] \proves \Gor{i \in I} \pass{\p}{\q}{\U[i]} \glob[i]
            }$
        }
        \\[6mm]
        & \multicolumn{1}{c}{
            $\tree {
                \set{\p[i]: \tvar{t} \setbar i \in I}, \set{\p[j]: \local[j] \setbar j \in J} \proves \glob
                \qquad
                \forall j \in J, \local[j] \not= \tvar{t}
            }{
                \set{\p[i]: \tinact \setbar i \in I}, \set{\p[j]: \trec{t} \local[j] \setbar j \in J} \proves \rec{t} \glob
            }$
        }
        \bstrut
        \\[4mm]
        \hline
        \hline
    \end{tabular}
}
\caption{Formal definitions for agent interaction semantics~\label{tab:formal_sessions}}
\end{table*}

Table~\ref{tab:formal_sessions} presents a table with the formal definitions for establishing the agent interaction semantics.

\subsection{Local protocols}
This section defines the theory for local protocols, which are textual descriptions of the commmunication interaction of a single agent.

Set
\[
\Participants = \set{\p, \q, \dots}
\]
is a set of interacting agents, such as controls, estimators, sensors and actuators, referred to as the set of {\em participants}.

Set 
\[
\Types = \set{\nat, \bool, \ON, \OFF, \dots}
\]
is a set of {\em message types}, 
e.g.,~integers (\nat), booleans (\bool),
enumerations labels (\ON, \OFF, \dots), etc.
Symbol \U denotes elements in \Types, i.e., $\U \in \Types$.

Set
\[
\alphabet = \Participants \times \set{\outputSymbol, \inputSymbol} \times \Types
\]
is an alphabet, whose elements,
$\element \in \alphabet$, describe participant {\em actions}.
The {\em send} action, \actOut{\p}{\U},
(written for convenience
instead of $(\p, \outputSymbol, \U)$)
describes the sending of a message with
type \U to participant \p. Dually,
the {\em receive} action, \actInp{\p}{\U} ,
describes the reception of a message of
type \U from participant \p.

Actions are used to type the send and receive
operators in a program. For example,
\[
\mathtt{send(p, 5)}: \actOut{\p}{\nat}
\]
denotes the typing of programming operator $\mathtt{send(p, 5)}$, which sends a value $5$ to a participant $\p$, with action \actOut{\p}{\nat}. Similarly,
\[
\mathtt{bool\ x = receive(p)}: \actInp{\p}{\bool}
\]
denotes the typing of programming operator $\mathtt{receive(p)}$, which receives a boolean value from a participant $\p$, with action \actInp{\p}{\bool}.

%
The {\em composition} of actions constructs the set of local protocols. 
Formally, the set of local protocols is inductively defined as:
\[
\begin{array}{rcl}
	\Local  &=&     \set{\element \sep \local \setbar \element \in \alphabet, \local \in \Local}
	\\      &\cup&  \set{\local[1] \tor \local[2] \setbar \local[1], \local[2] \in \Local}
	\\      &\cup&  \set{\trec{t} \local \setbar \local \in \Local}
	\\      &\cup&  \set{\tvar{t}, \tvar{t_1}, \dots}
	\\      &\cup&  \set{\tinact}
\end{array}
\]
\begin{enumerate}
	\item
	Local protocol $\element \sep \local$ denotes {\em sequential} composition, where a state \local is preceded by an action, $\element \in \alphabet$,
	i.e.,~a protocol in state $\element \sep \local$ observes
	action \element and proceeds to state \local.
	For example, local protocol
	\[
	\out{\p}{\nat} \inp{\q}{\bool} \local
	\]
	describes a program that first sends an integer value, type \nat, to participant \p, followed by sending a boolean value, type \bool, to participant \q. After these operations protocol \local describes the remaining program.
	\item
	Local protocol $\local[1] \tor \local[2]$ denotes a {\em choice} composition between local protocol \local[1] and local protocol \local[2]. For example, protocol:
	\[
	\out{\p}{\U[1]} \local[1] \tor \out{\p}{\U[2]} \local[2]
	\]
	declares a choice to send either a message of type \U[1] or a message type \U[2] to participant \p and then continue with protocol \local[1] or \local[2], respectively.
	
	Protocols of the form $\local[1] \tor \dots \tor \local[n]$, for some $n > 0$, are often abbreviated using notation $\Tor{1 \leq i \leq n} \local[i]$.
	\item
	Local protocol $\trec{t} \local$ declares a {\em recursive} protocol \local within a loop with label $\tvar{t}$, whereas label \tvar{t} denotes an execution jump to loop label $\tvar{t}$.
	For example, local protocol
	\[
	\trec{t} \inp{\p}{\U} \tvar{t}
	\]
	describes a reception of message \U from participant \p in a loop.
	\item   
	Local protocol \tinact is the {\em inactive} local protocol denoting the termination of a protocol.
	The inactive protocol is often omitted, 
	e.g., protocol $\out{\p}{\U[1]} \inp{\q}{\U[2]} \tinact$ is written as $\out{\p}{\U[1]} \ninp{\q}{\U[2]}$.
\end{enumerate}

The following example demonstrates the local protocols for the agents in \figMotivatingExample. 
\begin{example}[Sensor and estimator agents]
	\rm
	The local protocols for the flow sensors and the estimator in \figMotivatingExample 
	are defined, respectively, as:
	\begin{eqnarray*}
		\local[{\flowSensor}] &=& \trec{t} \out{\estpp}{\flowt} \tvar{t},
		\\
		\local[\estpp] &=& \trec{t} \inp{\flowSensor[1]}{\flowt} \inp{\flowSensor[2]}{\flowt} \out{\controller}{\headt} \tvar{t}.
	\end{eqnarray*}
	A flow sensor loops and sends a \flowt value to estimator \estpp. The estimator also loops and receives a \flowt value from the first sensor, \flowSensor[1], followed by a \flowt value from the second sensor, \flowSensor[2], and then estimates and sends a \headt value to the controller \controller.
	\qed
\end{example}

Function $\participants: \Local \to 2^{|\Participants|}$
returns the participants of a local protocol, inductively defined as:
\begin{eqnarray*}
	\participants[\tinact] &=& \es,
	\\
	\participants[\tvar{t}] &=& \es,
	\\
	\participants[\out{\p}{\U} \local] &=& \set{\p} \cup \participants[\local],
	\\
	\participants[\inp{\p}{\U} \local] &=& \set{\p} \cup \participants[\local],
	\\
	\participants[{\local[1] \tor \local[2]}] &=& \participants[{\local[1]}] \cup \participants[{\local[2]}],
	\\
	\participants[\trec{t} \local] &=& \participants[\local]
\end{eqnarray*}
For example,
$\participants[{\local[\estpp]}] = \set{\flowSensor[1], \flowSensor[2], \controller}$. 

Function $\size: \Local \to \mathbb{N}$ return the {\em syntactic size} of a local protocol, inductively defined as:
\begin{eqnarray*}
	\size[\tinact] &=& 1, 
	\\
	\size[\tvar{t}] &=& 1,
	\\
	\size[\out{\p}{\U} \local] &=& 1 + \size[\local],
	\\
	\size[\inp{\p}{\U} \local] &=& 1 + \size[\local],
	\\
	\size[{\local[1] \tor \local[2]}] &=& \size[{\local[1]}] + \size[{\local[2]}],
	\\
	\size[\trec{t} \local] &=& \size[\local]
\end{eqnarray*}

A {\em local transition} relation defines
the interaction semantics of a local protocol.
The definition of the transition relation,
requires the definition of
the {\em substitution} function. The substitution function
inputs a local protocol \local[1],
a local protocol \local[2], and
a loop label \tvar{t}, and replaces all
instances of \tvar{t} in \local[1] with \local[2]. 
The substitution function $[\substff]: \Local \times \Local \times \set{\tvar{t}, \tvar{t_1}, \dots} \to \Local$
is inductively defined as:

\begin{eqnarray*}
	\tinact \substt{\local}{\tvar{t}} &=& \tinact,
	\\[1mm]
	\tvar{t} \substt{\local}{\tvar{t}} &=& \local,
	\\[1mm]
	\tvar{t'} \substt{\local}{\tvar{t}} &=& \tvar{t'} \text{ if } \tvar{t} \not= \tvar{t'},
	\\[1mm]
	\element \sep \locald \substt{\local}{\tvar{t}} &=& \element \sep (\locald \substt{\local}{\tvar{t}}),
	\\[1mm]
	(\local[1] \tor \local[2]) \substt{\local}{\tvar{t}} &=& (\local[1] \substt{\local}{\tvar{t}}) \tor (\local[2] \substt{\local}{\tvar{t}}),
	\\[1mm]
	\trec{t'} \locald \substt{\local}{\tvar{t}} &=& 
	\left\{
	\begin{array}{ll}
		\trec{t'} (\locald \substt{\local}{\tvar{t}}) & \text{ if } \tvar{t} \not= \tvar{t'}
		\\
		\trec{t'} \locald & \text{ if } \tvar{t} = \tvar{t'}
	\end{array}
	\right.
\end{eqnarray*}

Substitution enables loop interactions, e.g., for type $\local[\flowSensor] = \trec{t} \out{\estpp}{\flowt} \tvar{t}$ substitution:
\[
\begin{array}{l}
	\out{\estpp}{\flowt} \tvar{t} \substt{\trec{t} \out{\estpp}{\flowt} \tvar{t}}{\tvar{t}} =
	\\ \hspace{4cm}
	\out{\estpp}{\flowt} \trec{t} \out{\estpp}{\flowt} \tvar{t}
\end{array}
\]
{\em unfolds} the body of the recursion of local protocol \local[\flowSensor].

A transition relation, $\by \subseteq \Local \times \alphabet \times \Local$,
is a set of triples
$(\local[1], \element, \local[2])$, where
state \local[1] {\em observes} action \element
and proceeds to state \local[2].
Triple $(\local[1], \element, \local[2]) \in \by$
is expressed in infix notation as $\local[1] \by[\element] \local[2]$.
The local transition relation is defined as:
\begin{itemize}
	\item   $\element \sep \local \by[\element] \local$, for all $\element \sep \local \in \Local$.
	\item   $\local[1] \tor \local[2] \by[\element] \locald$, if $\local[1] \by[\element] \locald$ or $\local[2] \by[\element] \locald$.
	\item   $\trec{t} \local \by[\element] \locald$, if $\local \substt{\trec{t} \local}{\tvar{t}} \by[\element] \locald$.
\end{itemize}

The following example demonstrates a loop transition.
\begin{example}[Loop transitions]
	\rm
	Consider local protocol $\local[\flowSensor] = \trec{t} \out{\estpp}{\flowt} \tvar{t}$.
	Building on the rules for the local transition, observe  that:
	\[
	\local[\flowSensor] \by[\actOut{\estpp}{\flowt}] \local[\flowSensor],
	\]
	since
	\begin{enumerate}[label=\roman*)]
		\item
		$\out{\estpp}{\flowt} \tvar{t} \substt{\local[\flowSensor]}{\tvar{t}} = \out{\estpp}{\flowt} \local[\flowSensor]$;
		and
		\item
		$\out{\estpp}{\flowt} \local[\flowSensor] \by[\actOut{\p}{\U}] \local[\flowSensor]$.
		\qed
	\end{enumerate}
\end{example}

\subsection{Local configurations}

A {\em role}, $\p: \local$, associates a participant, \p, with
a local protocol, \local. By extension, a {\em local configuration} is a mapping from participants to local protocols, i.e., a set of roles.
Set
\[
\LSet = \set{S \setbar S: \Participants \rightharpoonup \Local}
\]
is the set of mappings (partial functions)
from participants to local protocols.
The union of local configurations preserves
the partial function property;
$\LL[1], \LL[2] = \LL[1] \cup \LL[2]$ whenever
$\domain{\LL[1]} \cap \domain{\LL[2]} = \es$, and
undefined otherwise.

Function, $\activef: \LSet \to 2^{|\Participants|}$,
returns the {\em active} participants of a local configuration:
\[
\activep{\LL} = \domain{\LL} \backslash \set{\p \setbar \p: \tinact \in \LL}.
\]

Function $\size: \LSet \to \mathbb{N}$ returns the syntactic size of a local configuration:
\[
\size[\set{\p[i]: \local[i] \setbar i \in I}] = \sum\limits_{i \in I} \size[{\local[i]}]
\]
\begin{example}[A local configuration for the \WDNsimple]
\rm
\label{app:ex:local_configuration_running_example}
Consider the example in \figMotivatingExample 
and particularly the case
where the level sensor has failed and 
the control-loop deploys
a tank level estimator. The following
local configuration defines the
control loop configuration:
\[
\arraycolsep=2pt
\begin{array}{rcrl}
	\LLsimple   &=&     \estpp: & \trec{t} \inp{\flowSensor[1]}{\flowt} \inp{\flowSensor[2]}{\flowt} \out{\controller}{\headt} \tvar{t},
	\\ &&               \flowSensor[1]: & \trec{t} \out{\estpp}{\flowt} \tvar{t},
	\\ &&               \flowSensor[2]: & \trec{t} \out{\estpp}{\flowt} \tvar{t},
	\\ &&               \pumpp: & \trec{t} (\inp{\controller}{\ON} \tvar{t} \tor  \inp{\controller}{\OFF} \tvar{t}),
	\\ &&               \controller: & \trec{t} \inp{\estpp}{\headt} (\out{\pumpp}{\ON} \tvar{t} \tor \out{\pumpp}{\OFF} \tvar{t})
\end{array}
\]
%
Roles \flowSensor[1], \flowSensor[2],
output to \est,
the inflow and output values, respectively,
which in turn sends the estimated \headt value
to the controller, \controller. The controller then 
sends a binary signal, \ON or \OFF, to control
the pump.
%
%
%
\qed
\end{example}

Within a local configuration, roles {\em synchronise}
over their {\em dual} (send/receive actions)
interactions to define a {\em communication transition}
relation.
Alphabet \calphabet is the set of
all synchronisation actions:
\[
\calphabet = (\Participants \times \set{\passSymbol} \times \Participants \times \Types).
\]
The {\em message pass} action,
written as \act{\p}{\q}{\U} instead of $(\p, \passSymbol,\q, \U)$,
denotes the passing of message \U
from participant \p to participant \q.
%
A triple $\LL[1] \by[\element] \LL[2]$ denotes
that the agents in local configuration \LL[1] interact
to observe action, $\element \in \calphabet$, and
proceed to \LL[2].
The communication transition relation,
$\by \subseteq \LSet \times \calphabet \times \LSet$
is defined as:
\[
\begin{array}{l}
\LL, \p: \local[1], \q: \local[2] \by[\act{\p}{\q}{\U}] \LL, \p: \locald[1], \q: \locald[2],
\\
\qquad \qquad \qquad \qquad \text{ whenever }
\local[1] \by[\actOut{\q}{\U}] \locald[1]
\text{ and }
\local[2] \by[\actInp{\p}{\U}] \locald[2].
\end{array}
\]

It is desirable for a local configuration to satisfy the following properties:
i) The {\em deadlock-freedom} property requires that a local context can either perform an observable action or ensure all its roles remain inactive.
ii) The {\em liveness} property expects that every active role in a local configuration can eventually participate in an action.
The following definitions formally express the two properties:
\begin{itemize}
\item
A local context \LL is deadlock-free whenever:
\begin{itemize}
\item  
Either there exist $\LLd$ and $\element \in \calphabet$
such that $\LL \by[\element] \LLd$,
or $\activep{\LL} = \es$.
\item
For all \LLd, such that $\LL \by[\element] \LLd$ for some $\element \in \calphabet$, \LLd is deadlock-free.
\end{itemize}
\item
A local configuration \LL is live whenever:
\begin{itemize}
\item
For all $\p \in \activep{\LL}$,
there exist
$\LL[1], \dots, \LL[n], \LLd$
and
$\element[1], \dots, \element[n], \element \in \calphabet$, with $n \geq 0$,
such that
$\LL \by[{\element[1]}] \LL[1], \dots, \LL \by[{\element[n]}] \LL[n]$,
and
$\LL[n] \by[\element] \LLd$, where \element is of the form \act{\p}{\q}{\U} or \act{\q}{\p}{\U}.
\item   
For all \LLd, such that $\LL \by[\element] \LLd$ for some $\element \in \calphabet$, \LLd is live.
\end{itemize}
\end{itemize}

The following simple example demonstrates the deadlock-freedom and liveness properties.
\begin{example}
\rm
Consider local configurations:
\[
\arraycolsep=2pt
\begin{array}{rcrl}
\LL[1] &=& \sensor: & \trec{t} \out{\controller}{\flowt} \tvar{t},
\\     && \controller: & \trec{t} \inp{\sensor}{\flowt} \tvar{t}
\\[1mm]
\LL[2] &=& \estpp: & \nout{\controller}{\headt}
\end{array}
\]
Local configuration \LL[1] is deadlock-free and live.  In contrast, \LL[2] is neither deadlocked nor live.
Furthermore, the combined local configuration $\LL[3] = \LL[1], \LL[2]$ is deadlock-free but not live.
This verifies that deadlock-freedom does not imply liveness, since
\LL[3] can always observe transition
$\LL[3] \by[\act{\sensor}{\controller}{\flowt}] \LL[3]$
but role \estpp never participates in any action.
\end{example}

The following theorem states that a live local configuration is always deadlock-free.
\begin{theorem}
\rm
If \LL is live then \LL is deadlock-free. 
\end{theorem}
\begin{proof}
\rm
The proof proceeds by contrapositive, observing that a deadlocked process is not live.
\end{proof}

\subsection{Global protocols}

Liveness is a cornerstone property in distributed systems, ensuring safe interaction and progress for all components involved.
The mathematical definition of liveness requires verifying, for all active participants, the existence of valid transition paths and, moreover, recursively checking liveness for all reachable configurations. Algorithmically, a direct implementation of this definition necessitates exploring all possible transition paths of a local configuration. Consequently, this approach results in exponential time complexity relative to the syntactic size of the configuration.

The remainder of this section focuses on methods to ensure or construct live local configurations while avoiding procedures with exponential complexity.

{\em Global protocols} offer an alternative perspective on behavioural types, ensuring properties such as deadlock-freedom and liveness by design.
Formally, the set of global protocols is defined as:
\[
\arraycolsep=4pt
\begin{array}{rcl}
\Glob   &=&     \set{\element \sep \glob \setbar \element \in \calphabet, \glob \in \Glob}
\\      &\cup&  \set{\glob[1] \gor \glob[2] \setbar \glob[1], \glob[2] \in \Glob}
\\      &\cup&  \set{\rec{t} \glob \setbar \glob \in \Glob}
\\      &\cup&  \set{\var{t}, \var{t_1}, \dots}
\\      &\cup&  \set{\ginact}
\end{array}
\]
\begin{enumerate}
\item

A sequential global protocol is a protocol \glob prefixed by a message-passing action $\act{\p}{\q}{\U}$, expressed as $\pass{\p}{\q}{\U} \glob$. This describes participant \p sending a message of type \U to participant \q, after which the interaction proceeds according to protocol \glob.
\item
Global protocol $\glob[1] \gor \glob[2]$ is a choice between protocols $\glob[1]$ and $\glob[2]$. Notation $\Gor{1 \leq i \leq n} \glob[i]$ abbreviates protocol $\glob[1] \gor \dots \gor \glob[n]$, whenever $n > 0$.
\item
Global protocol $\rec{t} \glob$ declares a global protocol \glob within a loop with label \var{t}.
\item
The global protocol without any interaction is defined as \ginact and is often omitted.
\end{enumerate}
%

A global protocol specifies the behaviour of a local configuration by {\em projecting} its constituent roles. Conversely, a local configuration may {\em compose} a global protocol.
The projection/composition relationship is defined axiomatically using derivation trees.
Specifically, given a set of axioms in the form of derivation trees, a derivation tree:
\[
t = \tree{p_1, \dots, p_n}{p}
\]
with $n \leq 0$, is {\em derivable} from the axioms, thus
proposition $p$ is derivable from the axioms, whenever
either the derivation tree  $t$ is itself an axiom or the
propositions $p_1, \dots, p_n$ are derivable from the axioms.
\begin{definition}[Projection/Composition relation]
\label{app:def:composition_relation}
\rm
Relation,
$\proves \subseteq \LSet \times \Glob$,
is
defined as:
\[
\begin{array}{c}
\set{\p[i]: \tinact \setbar i \in I} \proves \ginact
\qquad
\set{\p[i]: \tvar{t} \setbar i \in I} \proves \var{t}
\\[2mm]
\tree{
	\LL, \p: \local, \q: \locald \proves \glob
}{
	\LL, \p: \out{\q}{\U} \local, \q: \inp{\p}{\U} \locald \proves \pass{\p}{\q}{\U} \glob
}
\\[6mm]
\tree{
	\forall \in I,\ \LL, \p: \local[i], \q: \locald[i] \proves \glob[i]
}{
	\LL, \p: \Tor{i \in I} \out{\q}{\U[i]} \local[i], \q: \Tor{i \in I} \inp{\p}{\U[i]} \locald[i] \proves \Gor{i \in I} \pass{\p}{\q}{\U[i]} \glob[i]
}
\\[8mm]
\tree {
	\set{\p[i]: \tvar{t} \setbar i \in I}, \set{\p[j]: \local[j] \setbar j \in J} \proves \glob
	\qquad
	\forall j \in J, \local[j] \not= \tvar{t}
}{
	\set{\p[i]: \tinact \setbar i \in I}, \set{\p[j]: \trec{t} \local[j] \setbar j \in J} \proves \rec{t} \glob
}
\end{array}
\]
\end{definition}

Projection/composition is defined
inductively on the syntax of global
protocols.
\begin{enumerate}
\item
An inactive local configuration composes an inactive global protocol.
\item
Similarly, a local configuration with local loop variables composes a global loop variable.
\item
The rule for composing the message passing action requires a send action by participant \p and a corresponding receive action by participant \q on message type \U.
\item
The rule for handling choice imposes restrictions on choice interaction. In particular, the rule allows composition for global protocols of the form $\Gor{i \in I} \pass{\p}{\q}{\U[i]} \glob[i]$, where participant \p chooses from a set of messages \U[i] to send to a participant \q.
Moreover, the rule requires that all roles besides \p and \q implement the same local protocol in each continuation $\glob[i]$, as shown by requirement for local configuration \LL in condition:
\[
\forall i \in I,\ \LL, \p: \local[i], \q: \locald[i] \proves \glob[i].
\]
\item
Finally, the recursive global type is composed by a local configuration of recursive local types, all prefixed with the same local loop variable.
\end{enumerate}
The restrictions on the choice rule ensure the liveness property for local configurations.
\begin{theorem}
\label{app:thm:global_is_live}
If $\LL \proves \glob$, then \LL is live. 
\end{theorem}
\begin{proof}
The proof is done by induction on the definition of $\proves$.

\begin{enumerate}[label=$\bullet$]
\item   \textbf{Base Case:}
The base case, 
\[
\set{\p[i]: \tinact \setbar i \in I} \proves \ginact
\]
is straightforward. For local configurations $\LL = \set{\p[i]: \tinact \setbar i \in I}$,
\LL is live.

\item       \textbf{Inductive Hypothesis:} Assume that if 
\[
\LL \proves \glob
\]
then \LL is live. 

\item       \textbf{Inductive Step:}
There are three cases: 
\begin{enumerate}
	\item   The case for the choice global protocol:
	\[
	{
		\scriptstyle
		\tree{
			\forall \in I,\ \LL, \p: \local[i], \q: \locald[i] \proves \glob[i]
		}{
			\scriptstyle
			\LL, \p: \Tor{i \in I} \out{\q}{\U[i]} \local[i], \q: \Tor{i \in I} \inp{\p}{\U[i]} \locald[i] \proves \Gor{i \in I} \pass{\p}{\q}{\U[i]} \glob[i]
		}
	}
	\]
	There are two sub-cases: 
	
	\noindent
	\begin{enumerate}
		\item   Observe that for all $i \in I$ 
		\[
		\begin{array}{cl}
			& \LL, \p: \Tor{i \in I} \out{\q}{\U[i]} \local[i], \q: \Tor{i \in I} \inp{\p}{\U[i]} \locald[i]
			\\ \by[\act{\p}{\q}{\U[i]}] &
			\LL, \p:  \local[i], \q:  \locald[i]
		\end{array}
		\]
		By the induction hypothesis, \LL, \p:  \local[i], \q:  \locald[i] is live.
		
		\item   Observe that
		\[
		\begin{array}{cl}
			&\LL, \p: \Tor{i \in I} \out{\q}{\U[i]} \local[i], \q: \Tor{i \in I} \inp{\p}{\U[i]} \locald[i]
			\\ \by[\act{\rr[1]}{\rr[2]}{\U}] &
			\LLd, \p: \Tor{i \in I} \out{\q}{\U[i]} \local[i], \q: \Tor{i \in I} \inp{\p}{\U[i]} \locald[i]
		\end{array}
		\]
		with $\LL \by[\act{\rr[1]}{\rr[2]}{\U}]  \LLd$.
		By the inductive hypothesis it holds that for all $i \in I$,
		\[
		\LL, \p: \local[i], \q: \locald[i]
		\by[\act{\rr[1]}{\rr[2]}{\U}]
		\LLd, \p: \local[i], \q: \locald[i]
		\]
		implies that $\LLd, \p: \local[i], \q: \locald[i]$ is live.
		From the above results and the definition of liveness, it follows that:
		\[
		\LLd, \p: \Tor{i \in I} \out{\q}{\U[i]} \local[i], \q: \Tor{i \in I} \inp{\p}{\U[i]} \locald[i]
		\] is live.
	\end{enumerate}
	\item   The case for message passing global protocol:
	\[
	\tree{
		\LL, \p: \local, \q: \locald \proves \glob
	}{
		\LL, \p: \out{\q}{\U} \local, \q: \inp{\p}{\U} \locald \proves \pass{\p}{\q}{\U} \glob
	}
	\]
	This is a special case of the choice global protocol rule when the set $I$ is a singleton.
	\item   The case for recursive global protocol:
	\[
	\tree {
		\scriptstyle
		\set{\p[i]: {\color{olive} \mathsf{t}} \setbar i \in I}, \set{\p[j]: \local[j] \setbar j \in J} \proves \glob
		\qquad
		\forall j \in J, \local[j] \not= {\color{olive} \mathsf{t}}
	}{
		\scriptstyle
		\set{\p[i]: \tinact \setbar i \in I}, \set{\p[j]: {\color{olive} \mathsf{loop\;t}}. \local[j] \setbar j \in J} \proves {\color{blue} \mathsf{loop\;t}}. \glob
	}
	\]
	is straightforward. 
	A local configuration:
	\[
	\set{\p[i]: \tinact \setbar i \in I}, \set{\p[j]: \trec{t} \local[j] \setbar j \in J}
	\]
	is semantically equivalent, i.e., has the same transition communication transition, to
	\[
	\set{\p[i]: \tinact \setbar i \in I}, \set{\p[j]: \local[j] \substt{\trec{t} \local[j]}{\tvar{t}} \setbar j \in J}
	\]
	Moreover, global protocol
	\[
	\rec{t} \glob
	\]
	is equivalent to
	$\glob \substt{\rec{t} \glob}{\var{t}}$.
	
	Applying the inductive hypothesis gives
	$
	\set{\p[i]: \tinact \setbar i \in I}, \set{\p[j]: \local[j] \substt{\trec{t} \local[j]}{\tvar{t}} \setbar j \in J} \proves \glob \substt{\rec{t} \glob}{\var{t}}
	$,
	and thus:
	\[
	\set{\p[i]: \tinact \setbar i \in I}, \set{\p[j]: \local[j] \substt{\trec{t} \local[j]}{\tvar{t}} \setbar j \in J}
	\]
	is live. 
\end{enumerate}
\end{enumerate}
The inductive step concludes the proof.
\end{proof}

The composition/projection algorithm has polynomial time complexity relative to the syntactic size of a local configuration.
\begin{theorem}
\label{app:thm:project_is_polynomial}
$\LL \proves \glob \in O(\size[\LL])$.
\end{theorem}

\begin{proof}
The proof is done by induction on the definition of $\proves$.

\begin{enumerate}[label=$\bullet$]
\item   \textbf{Base Case:}
There are two cases
\begin{itemize}
	\item
	Base case $\set{\p[i]: \tinact \setbar i \in I} \proves \ginact$ is straightforward. The algorithm requires \size[I] steps, one for each $\p[i]$, to check that the corresponding local protocol is inactive.
	Moreover, 
	\[
	\textstyle
	\size[\set{\p[i]: \tinact \setbar i \in I}] = \sum\limits_{i \in I} \size[\tinact] = \sum\limits_{i \in I} 1 = \size[I].
	\]
	Thus, \[
	\set{\p[i]: \tinact \setbar i \in I} \proves \ginact \in O(\size[\set{\p[i]: \tinact \setbar i \in I}]),
	\]
	as required.
	
	\item
	Base case $\set{\p[i]: \tvar{t} \setbar i \in I} \proves \var{t}$ follows similar argumentation.
\end{itemize}

\item       \textbf{Inductive Hypothesis:} Assume that $\LL \proves \glob \in O(\size[\LL])$.

\item       \textbf{Inductive Step:}
There are three cases: 
\begin{enumerate}
	\item   The case for message passing global protocol:
	\[
	\tree{
		\LL, \p: \local, \q: \locald \proves \glob
	}{
		\LL, \p: \out{\q}{\U} \local, \q: \inp{\p}{\U} \locald \proves \pass{\p}{\q}{\U} \glob
	}
	\]
	The rule takes one step to match the duality of actions \actOut{\q}{\U} and \actInp{\p}{U} and then verifies
	$\LL, \p: \local, \q: \locald \proves \glob$.
	By the inductive hypothesis:
	\[
	\LL, \p: \local, \q: \locald \proves \glob \in O(\size[\LL, \p: \local, \q: \locald]).
	\]
	Thus, it is safe to conclude that:
	\[
	\begin{array}{l}
		\LL, \p: \out{\q}{\U} \local, \q: \inp{\p}{\U} \locald \proves \pass{\p}{\q}{\U} \glob \qquad \qquad \quad
		\\ \hfill \in O(2 + \size[\LL, \p: \local, \q: \locald]).
	\end{array}
	\]
	Furthermore, observe that 
	\[
	\size[\LL, \p: \out{\q}{\U} \local, \q: \inp{\p}{\U} \locald] = 2 + \size[\LL, \p: \local, \q: \locald],
	\]
	concluding with:
	\[
	\begin{array}{l}
		\LL, \p: \out{\q}{\U} \local, \q: \inp{\p}{\U} \locald \proves \pass{\p}{\q}{\U} \glob \qquad \qquad \qquad
		\\ \hfill \in O(\size[\LL, \p: \out{\q}{\U} \local, \q: \inp{\p}{\U} \locald]).
	\end{array}
	\]
	\item   The case for the choice global protocol:
	\[
	{
		\scriptstyle
		\tree{
			\forall \in I,\ \LL, \p: \local[i], \q: \locald[i] \proves \glob[i]
		}{
			\scriptstyle
			\LL, \p: \Tor{i \in I} \out{\q}{\U[i]} \local[i], \q: \Tor{i \in I} \inp{\p}{\U[i]} \locald[i] \proves \Gor{i \in I} \pass{\p}{\q}{\U[i]} \glob[i]
		}
	}
	\]
	The rule matches the duality of actions \actOut{\q}{\U[i]} and \actInp{\p}{\U[i]}, requiring \size[I] steps, and then checks
	$\LL, \p: \local[i], \q: \locald[i] \proves \glob[i]$, for all $i \in I$.
	From the inductive hypothesis
	\[
	\forall i \in I,\ \LL, \p: \local[i], \q: \locald[i] \proves \glob[i] \in O(\size[{\LL, \p: \local[i], \q: \locald[i]}])
	\]
	Adding the steps together gives:
	\[
	\begin{array}{l}
		\size[{\LL, \p: \Tor{i \in I} \out{\q}{\U[i]} \local[i], \q: \Tor{i \in I} \inp{\p}{\U[i]} \locald[i]}] \qquad \qquad \qquad \qquad
		\\ \hfill 
		\begin{array}{l}
			= 2*\size[I] + \sum\limits_{i \in I} \size[{\LL, \p: \local[i], \q: \locald[i]}]
			\\
			= \sum\limits_{i \in I} (2 + \size[{\LL, \p: \local[i], \q: \locald[i]}]),
		\end{array}
	\end{array}
	\]
	which concludes with:
	\[
	\begin{array}{cl}
		\scriptstyle \LL, \p: \Tor{i \in I} \out{\q}{\U[i]} \local[i], \q: \Tor{i \in I} \inp{\p}{\U[i]} \locald[i] \proves \Gor{i \in I} \pass{\p}{\q}{\U[i]} \glob[i] \qquad \qquad 
		\\
		\hfill \in\ \scriptstyle O(\size[{\LL, \p: \Tor{i \in I} \out{\q}{\U[i]} \local[i], \q: \Tor{i \in I} \inp{\p}{\U[i]} \locald[i]}]).
	\end{array}
	\]
	
	\item
	The case for recursive global protocol:
	\[
	\tree {
		\scriptstyle
		\set{\p[i]: {\color{olive} \mathsf{t}} \setbar i \in I}, \set{\p[j]: \local[j] \setbar j \in J} \proves \glob
		\qquad
		\forall j \in J, \local[j] \not= {\color{olive} \mathsf{t}}
	}{
		\scriptstyle
		\set{\p[i]: \tinact \setbar i \in I}, \set{\p[j]: {\color{olive} \mathsf{loop\;t}}. \local[j] \setbar j \in J} \proves {\color{blue} \mathsf{loop\;t}}. \glob
	}
	\]
	is straightforward. The rule proceeds by checking 
	\[
	\set{\p[i]: \tvar{t} \setbar i \in I}, \set{\p[j]: \local[j] \setbar j \in J} \proves \glob.
	\]
	The inductive hypothesis ensures,
	\[
	\begin{array}{l}
		\set{\p[i]: \tvar{t} \setbar i \in I}, \set{\p[j]: \local[j] \setbar j \in J} \proves \glob \qquad \qquad \qquad
		\\
		\hfill \in O(\size[\set{\p[i]: \tvar{t} \setbar i \in I}, \set{\p[j]: \local[j] \setbar j \in J}]).
	\end{array}
	\]
	Moreover,
	\[
	\begin{array}{cl}
		\size[\set{\p[i]: \tinact \setbar i \in I}, \set{\p[j]: \trec{t} \local[j] \setbar j \in J}] \qquad \qquad
		\\
		\hfill = \size[\set{\p[i]: \tvar{t} \setbar i \in I}, \set{\p[j]: \local[j] \setbar j \in J}],
	\end{array}
	\]
	concluding that 
	\[
	\begin{array}{l}
		\scriptstyle \set{\p[i]: \tinact \setbar i \in I}, \set{\p[j]: {\color{olive} \mathsf{loop\;t}}. \local[j] \setbar j \in J} \proves {\color{blue} \mathsf{loop\;t}}. \glob \qquad \qquad \qquad 
		\\
		\hfill
		\in\ \scriptstyle O(\size[\set{\p[i]: \tinact \setbar i \in I}, \set{\p[j]: {\color{olive} \mathsf{loop\;t}}. \local[j] \setbar j \in J} \proves {\color{blue} \mathsf{loop\;t}}. \glob]).
	\end{array}
	\]
\end{enumerate}
\end{enumerate}
The inductive step concludes the proof.
\end{proof}

The following example presents a global protocol for the \WDNsimple use case.

\begin{example} 
\rm
\label{app:ex:globsimple}
Consider global protocol:
\[
\arraycolsep=3pt
\begin{array}{rcll}
\globsimple     &=& 
\rec{t} & \pass{\flowSensor[1]}{\estpp}{\flowt}
\\ &&&              \pass{\flowSensor[2]}{\estpp}{\flowt}
\\ &&&              \pass{\estpp}{\controller}{\headt} 
\\ &&&              (\pass{\controller}{\pumpp}{\ON}\var{t} \tor \pass{\controller}{\pumpp}{\OFF} \var{t})
\end{array}
\]
The example describes the control loop with tank water level estimator in \figMotivatingExample 
from a global point of view. The control loop begins with sensors \flowSensor[1] and \flowSensor[2] sending flow value to the estimator \estpp. The estimator then computes a water level as a \head value and sends it to the controller, \controller. Finally, the controller interacts with the pump, \pumpp, sending a signal to turn the pump \ON or \OFF.

It is easy to verify that $\LLsimple \proves \globsimple$,
where \LLsimple is the local configuration in Example~\ref{app:ex:local_configuration_running_example}, and thus verify that \LLsimple is live following Theorem~\ref{app:thm:global_is_live}.
\qed
\end{example}

\section{An ontology for the Water Distribution Network Domain}
\label{app:sec:wdn_ontology}
An ontology schema for the Water Distribution Networks (WDN) domain is defined by the structure:
\[
\WDN = \lrangle{\Properties, \emodel, \Classes, \translation}.
\]
The set of properties, $\Properties = \set{\flow, \head}$, defines the flow and the head properties. 
The set of estimators, $\emodel = \set{\tankmass, \junctionmass, \energy}$, defines the tank mass estimator, \tankmass, and the junction mass estimator, \junctionmass, adhering to the law of mass preservation, and the energy preservation estimator, \energy, adhering to the law of energy preservation.

The set of industrial component classes is defined as:
\[
\begin{array}{rcll}
\Classes &=&    \set{
&   \lrangle{r, \set{\head}},
\\ & &          &   \lrangle{j, \set{\flow, \head}, \junctionmass},
\\ & &          &   \lrangle{t, \set{\shape, \head}, \tankmass},
\\ & &          &   \lrangle{p, \set{\shape, \flow}, \energy},
\\ & &          &   \lrangle{u, \set{\shape, \flow}, \energy},
\\ & &          &   \lrangle{\class[\mathsf{flow}], \set{\flow}},
\\ & &          &   \lrangle{\class[\mathsf{head}], \set{\head}},
\\ & &          }
\end{array}
\]
The first five classes correspond to the following physical components: reservoirs, junctions, tanks, pipes, and pumps.
The last two classes represent flow sensing points and head sensing points.
The pump class is the sole member of the set of actuator classes.
Additionally, reservoirs, junctions, and tanks are categorised as physical node classes, while pipes and pumps are classified as physical link classes.

A knowledge graph for a water distribution network process, \wdnprocess, is defined as:
\[
\wdnprocess = \lrangle{\GG, \lrangle{\Hardware, \h}, \lrangle{\AG, \modelmap}}
\]
Graph $\GG = \lrangle{\VV, \EE}$, is called {\em water distribution network graph}.
Following the definition of industrial component classes, set $\VV$, with $v_1, v_2, \dots \in \VV$, is partitioned into:
i) reservoirs, $\VV[r]$;
ii) junctions, $\VV[j]$;
iii) tanks, $\VV[t]$;
iv) pipes, $\VV[p]$;
v) pumps, $\VV[u]$;
vii) flow sensing points, \Sensors[\mathsf{flow}]; and 
viii) head sensing points, \Sensors[\mathsf{head}].
%
%
%
Structure \lrangle{\Hardware, \h} follows the general definition of industrial process ontologies.

\begin{figure}

    \begin{tikzpicture}[scale=0.9, transform shape]
        \reservoirGraph{0.5}{1}{Reservoir}{\reservoir}
        \junctionGraph{2.5}{1}{Junction}{\junction}
        \tankGraph{4.5}{1}{Tank}{\tank}

        \pipeGraph{6.5}{1}{Pipe}{\link}
        \node[head, pstate]        at  (6, 1.25) (ipr) {};
        \node[head, pstate]        at  (8, 1.25) (opr) {};
    
        \draw                       (ipr) edge[post] (agPipe)
                                    (opr) edge[post] (agPipe);
    
        \node[flow, pstate]         at (0, 0.5) (pr) {};
        \node[right=1mm]            at (pr) {\mysf{flow\ state}};

        \node[tmass, estimator]     at (2.75, 0.5) (flowPres) {};
        \node[right=1mm]            at (flowPres) {\mysf{tmass\ estimator}};
    
        \node[jmass, estimator]     at (6, 0.5) (dem)  {};
        \node[right=1mm]            at (dem) {\mysf{jmass\ estimator}};

        \node[head, pstate]         at (0, 0) (fl) {};
        \node[right=1mm]            at (fl) {\mysf{head\ state}};

        \node[energy, estimator]    at (2.75, 0) (fl) {};
        \node[right=1mm]            at (fl) {\mysf{energy\ estimator}};

       \node[constant, pstate]
                                    at (6, 0) (shape) {};
       \node[right=1mm]             at (shape) {\mysf{shape\ coef}};
    \end{tikzpicture}

    \caption{Analytical redundancy subgraphs for the Water Distribution Network ontology.
    Pentagon shapes represent state nodes, and square shapes represent estimator nodes. 
    \label{fig:components}}
\end{figure}

The state estimation translation function:
i) translates each physical component, $v \in \VV$, into a state estimation subgraph expressing state estimation within the component class;
ii) uses the information from industrial process graph edges $(v_1, v_2) \in \EE$, to interlink these subgraphs, forming a comprehensive state estimation graph.

The state estimation translation function operates as follows:
It maps each physical component $v \in \VV$ into a state estimation subgraph. This subgraph represents state estimation within the corresponding component class. Information from the industrial process graph edges $(v_1, v_2) \in \EE$ interlinks these subgraphs, creating a comprehensive state estimation graph. Fig.~\ref{fig:components} shows the physical component classes with their attributes and the translation of each class in the water distribution network into its internal subgraph.

A reservoir has a hydraulic \head state.
A junction has a \head state and a \flow state, denoting the junction demand. It includes a junction mass estimator, \junctionmass, that enforces mass preservation (the equivalence of inflow and outflow).
A tank has a \shape state and a \head state. The tank mass estimator, \tankmass, adheres to the mass preservation principle. The initial condition, along with the tank's inflow and outflow, determines water storage and, consequently, the tank's \head state.
Link nodes have a \shape state and a \flow state. An energy estimator, \energy, models the preservation of flow and head across the link. It uses the head state difference at the link's edges as input and determines the flow state within the link.

Given a water distribution network graph, $\GG = \lrangle{\VV, \EE}$, the state estimation translation function for the \WDN domain is defined as
$
\translation[\lrangle{\VV, \EE}] = \lrangle{\eV, \eE}
$
where
\[
\arraycolsep=2pt
\begin{array}{rcl}
\eV  &=&    \Sensors
\\  &\cup&  \set{v.\head \setbar v \in \VV[r]}
\\  &\cup&  \set{v.\head, v.\flow, v.\junctionmass \setbar v \in \VV[j]}
\\  &\cup&  \set{v.\shape, v.\head, v.\tankmass \setbar v \in \VV[t]}
\\  &\cup&  \set{v.\shape, v.\flow, v.\energy \setbar v \in \VV[l]}
\\[2mm]
\eE &=&     \set{(v.\shape, v.\tankmass) \setbar v \in \VV[t]}
\\ &\cup&   \set{(v.\shape, v.\energy) \setbar v \in \VV[l]}
\\ &\cup&   \set{(v_{l}.\flow, v_{c}.\mass), (v_{c}.\mass, v_{l}.\flow),

\\ && \quad (v_{c}.\head, v_{l}.\energy) \setbar (v_{l}, v_{c}) \in \EE \vee (v_{c}, v_{l}) \in \EE}
\\  &\cup&  \set{(s_p, v.\head) \setbar (v, s_p) \in \EE}
\\  &\cup&  \set{(s_f, v_e.\flow) \setbar (v, s_f) \in \EE}
\end{array}
\]
The connections in the state estimation graph define the dependencies among states and the estimation functions, as well as the states measured at sensing points.

The agent repository consists of a set of roles and the agent mapping. Agent roles are defined by behaviours that process input states (measured or estimated) and output states (measured or estimated) according to the estimation model and sensing points.
Agent roles are defined as:
$\AG = \AG[s] \cup \AG[g] \cup \AG[a]$, where
\[
\arraycolsep=2pt
\begin{array}{rcl}
\AG[s] &=& \levelSensor: \localhead,\ \flowSensor: \localflow
\\
\AG[g] &=& \junctionpp: \localjun,\ \linkpp: \locallink,\ \estpp: \localest, , 
\\
\AG[a] &=& \pumpp: \trec{t} (\inp{\producerp}{\ON} \tvar{t} \tor \inp{\producerp}{\OFF} \tvar{t})
\end{array}
\]
with    
\[
\arraycolsep=2pt
\begin{array}{rcl}
\localhead &=& \trec{t} \out{\consumerp}{\headt} \tvar{t},
\\[1mm]
\localflow &=& \trec{t} \out{\consumerp}{\flowt} \tvar{t},
\\[1mm]
\localjun &=&
\trec{t} \inp{\producerp[1]}{\flowt} \inp{\producerp[2]}{\flowt} \out{\consumerp}{\flowt} \tvar{t},
\\[1mm]
\locallink &=&
\trec{t} \inp{\producerp[1]}{\headt} \inp{\producerp[2]}{\headt} \out{\consumerp}{\flowt} \tvar{t},
\\[1mm]
\localest &=&
\trec{t} \inp{\producerp[1]}{\flowt} \inp{\producerp[2]}{\flowt} \out{\consumerp}{\headt} \tvar{t},
\\[1mm]
\localpump &=& \trec{t} (\inp{\producerp}{\ON} \tvar{t} \tor \inp{\producerp}{\OFF} \tvar{t})
\end{array}
\]
There are three estimator roles.
Role $\junctionpp:\localjun$ receives inflow and demand from \producerp[1] and \producerp[2], respectively, and sends the outflow estimation to \consumerp.
Role $\estpp:\localest$ receives inflow and outflow from \producerp[1] and \producerp[2], respectively, and sends tank head estimations to \consumerp.
Role $\linkpp:\locallink$ receives head measurements at the edges of a link from \producerp[1] and \producerp[2], respectively, and sends link flow estimations to \consumerp.
Note that mass estimators handle cases with two flow inputs. This specification remains general, as defining multiple estimators can handle systems with more flow inputs.
Additionally, there are two sensor roles:
Role $\flowSensor:\localflow$ measures and outputs flow states.
Role $\levelSensor:\localhead$ measures and outputs head states.
Finally, the actuator role $\pumpp:\localpump$ describes pump interaction where a pump can receive a binary signal to turn on or off the pump.

The agent mapping is defined as:
\[
\begin{array}{rcll}
\modelmap &=&   \set{ &
\junctionpp: \localjun \maparrow \junctionmass,
\\ && &             \estpp: \localest \maparrow \tankmass,          
\\ && &             \linkpp: \locallink \maparrow \energy
\\ && &             \flowSensor: \localflow \maparrow \flow,
\\ &&&              \levelSensor \localhead \maparrow \head,
\\ && &             \pumpp: \localpump \maparrow u
\hfill            }
\end{array}
\]
which maps:
estimator agents ($\junctionpp: \localjun$, $\estpp: \localest$, $\linkpp: \locallink$) to the estimator functions $\junctionmass$, $\tankmass$, and $\energy$, respectively;
sensor agents ($\flowSensor: \localflow$, $\levelSensor: \localhead$) to the properties $\flow$ and $\head$, respectively;
and actuator agent ($\pumpp: \localpump$) to the pump node $u$;
Estimator roles do not process input for \shape states, as agents can be preconfigured with the physical shape of their components before deployment.

\begin{figure}
\begin{tabular}{c}
   \begin{tikzpicture}[x=10mm, y=10mm]
        \begin{scope}[scale=0.6]
            \node[junction]     at  (0, 2.5) (j1) {};
            \node[above]        at  (j1) {$j_1$};
        
            \node[pump]         at  (2, 2.5) (pmp) {}
                                    edge[pre] (j1);
            \node[above]        at  (pmp) {\pump};
        
            \node[junction]     at  (4, 2.5) (j2) {}
                                    edge[pre] (pmp);
            \node[above]        at  (j2) {$j_2$};
        
            \node[pipe]         at  (6, 2.5) (p1) {}
                                    edge[pre] (j2);
            \node[above]        at  (p1) {$p_1$};
        
            \node[tank]         at  (8, 2.5) (tank) {}
                                    edge[pre] (p1);
            \node[above]        at  (tank) {$\tank$};
        
            \node[pipe]         at  (10, 2.5) (p2) {}
                                    edge[pre] (tank);
            \node[above]        at  (p2) {$p_2$};
        
            \node[junction]     at  (12, 2.5) (j3) {}
                                    edge[pre] (p2);
            \node[above]        at  (j3) {$j_3$};
                                        
        
            \node[spoint, head]
                                    at  (0, 1.5) (s1) {}
                                        edge[pre] (j1);
            \node[right]             at  (s1) {$s_1$};
        
            \node[spoint, flow]
                                    at  (2, 1.5) (s2) {}
                                        edge[pre] (pmp);
            \node[right]            at  (s2) {$s_2$};
        
            \node[spoint, head]
                                    at  (4, 1.5) (s3) {}
                                        edge[pre] (j2);
            \node[left]             at  (s3) {$s_3$};
        
            \node[spoint, flow]
                                    at  (4, 0.5) (s4) {}
                                        edge[pre, bend right] (j2);
            \node[right]            at  (s4) {$s_4$};
            
            \node[spoint, flow]
                                    at  (6, 1.5) (s5) {}
                                        edge[pre] (p1);
            \node[right]            at  (s5) {$s_5$};
        
            \node[spoint, head]
                                    at  (8, 1.5) (s6) {}
                                        edge[pre] (tank);
            \node[right]            at  (s6) {$s_6$};
        
            \node[spoint, flow]
                                    at  (10, 1.5) (s7) {}
                                        edge[pre] (p2);
            \node[right]            at  (s7) {$s_7$};
        
            \node[spoint, flow]
                                    at  (12, 1.5) (s8) {}
                                        edge[pre] (j3);
            \node[right]            at  (s8) {$s_8$};
        
        
            \node[plc,minimum width=0.75cm, minimum height=0.75cm]  at  (2, 0) (plc1) {}
                                        edge[pre]   (s1)
                                        edge[pre]   (s2)
                                        edge[pre]   (s3)
                                        edge[pre]   (s4)
                                        edge[pre, bend left]  (pmp);
            \node                 at    (plc1) {$\dev[1]$};
            
            \node[plc,minimum width=0.75cm, minimum height=0.75cm]  at  (8, 0) (plc2) {}
                                        edge[pre] (s5)
                                        edge[pre] (s6)
                                        edge[pre] (s7);
            \node                 at    (plc2) {$\dev[2]$};
        
            \node[plc,minimum width=0.75cm, minimum height=0.75cm]  at  (12, 0) (plc3) {}
                                        edge[pre] (s8);
            \node                 at    (plc3) {$\dev[3]$};

            \draw[dotted, <->]          (plc1) edge (plc2);
            \draw[dotted, <->]          (plc2) edge (plc3);
        
        \end{scope}
\end{tikzpicture}
\\[2mm]
\begin{tikzpicture}[x=10mm, y=10mm]
        \begin{scope}[scale=0.6]
            \junctionGraph{-0.5}{0}{j1}{}
            \pumpGraph{1.5}{0}{pump}{}
        
            \draw           (agpump) edge[pre] (prj1)
                            (flj1) edge[<->] (flpump);
        
            \junctionGraph{3.5}{0}{j2}{}
        
            \draw           (agpump) edge[pre] (prj2)
                            (flpump) edge[<->] (flj2);
        
            \pipeGraph{5.5}{0}{p1}{} 
        
            \draw           (flj2) edge[<->] (flp1)
                            (prj2) edge[post] (agp1);
        
            \tankGraph{7.5}{0}{tank}{} 
        
            \draw           (agtank) edge[<->] (flp1)
                            (prtank) edge[post] (agp1);
            
            \pipeGraph{9.5}{0}{p2}{} 
            \draw           (agtank) edge[<->] (flp2)
                            (prtank) edge[post] (agp2);
            
            \junctionGraph{11.5}{0}{j3}{} 
            \draw           (flp2) edge[<->] (flj3)
                            (agp2) edge[pre] (prj3);
                            
        
            \node[spoint, head]
                            at  (0, 3) (s1) {}
                                edge[post, bend right] (prj1);
            \node[right]    at  (s1) {$s_1$};
        
            \node[spoint, flow]
                            at  (2, 3) (s2) {}
                                edge[post] (flpump);
            \node[right]    at  (s2) {$s_2$};
        
            \node[spoint, head]
                            at  (4, 3) (s3) {}
                                edge[post, bend right] (prj2);
            \node[left]     at  (s3) {$s_3$};
        
            \node[spoint, flow]
                            at  (4.75, 2.5) (s4) {}
                                edge[post, bend left] (demj2);
            \node[right]    at  (s4) {$s_4$};
            
            \node[spoint, flow]
                            at  (6, 3) (s5) {}
                                edge[post] (flp1);
            \node[right]    at  (s5) {$s_5$};
        
            \node[spoint, head]
                            at  (8, 3) (s6) {}
                                edge[post, bend left] (prtank);
            \node[right]    at  (s6) {$s_6$};
        
            \node[spoint, flow]
                            at  (10, 3) (s7) {}
                                edge[post] (flp2);
            \node[left]     at  (s7) {$s_7$};
        
            \node[spoint, flow]
                            at  (12, 3) (s8) {}
                                edge[post, bend left] (demj3);
            \node[right]    at  (s8) {$s_8$};
        
            
            %
            %
            %
            %
            %
            %
        \end{scope}
\end{tikzpicture}
\\[2mm]
\input{tikz/WDN_control_configuration}
\end{tabular}
\caption{
Top: Graphical representation of the water distribution network, \WDNsimple,
of the running example in \figMotivatingExample. 
Middle: Graphical representation of the
analytical redundancy graph of \WDNsimple.
Bottom: Graphical representation of
seven analytical redundancy trees,
$\treesimple[1] \cup \dots \cup \treesimple[7]$
(shape coefficients are not depicted) rooted at $\tank.\head$,
with a hardware device assignment for each node. 
\label{app:fig:WDN_IP}
}
\end{figure}

\begin{figure*}
\begin{center}
\begin{tabular}{c}
	\fbox{\includegraphics[scale=0.5]{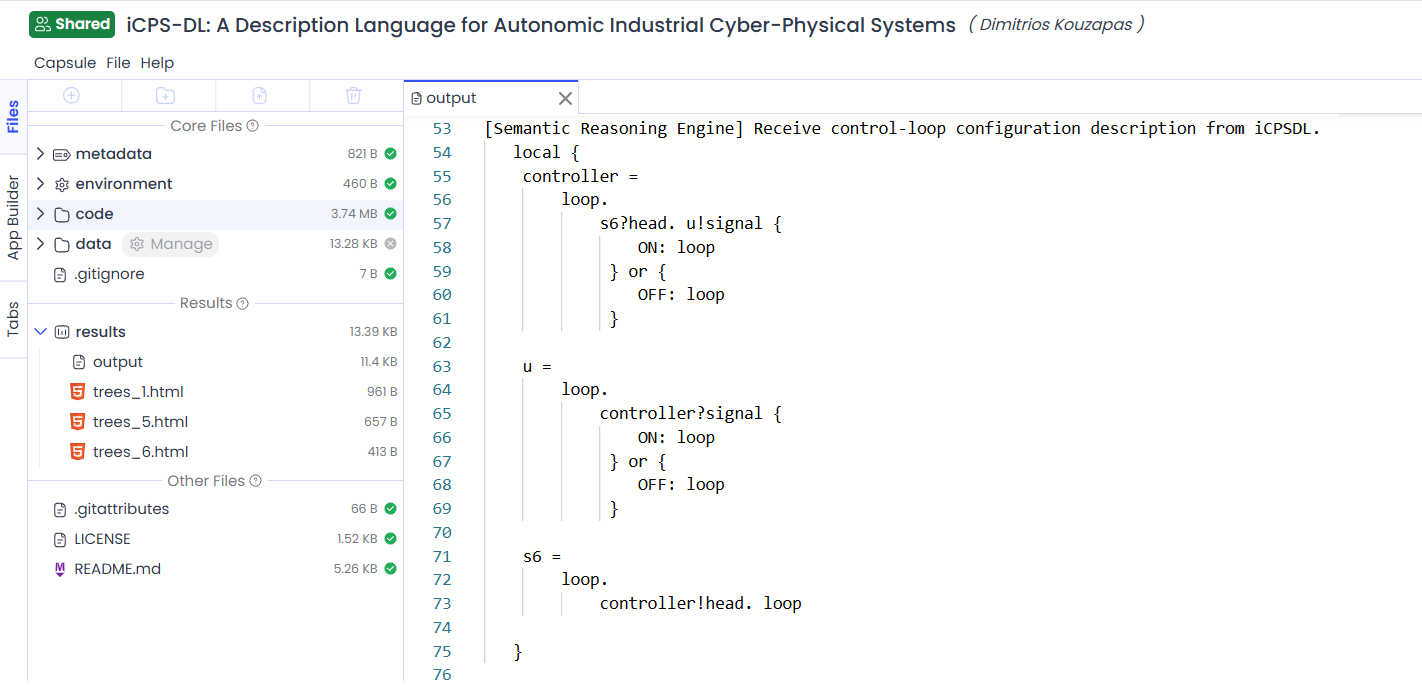}}
	\\[4mm]
	\fbox{\includegraphics[scale=0.5]{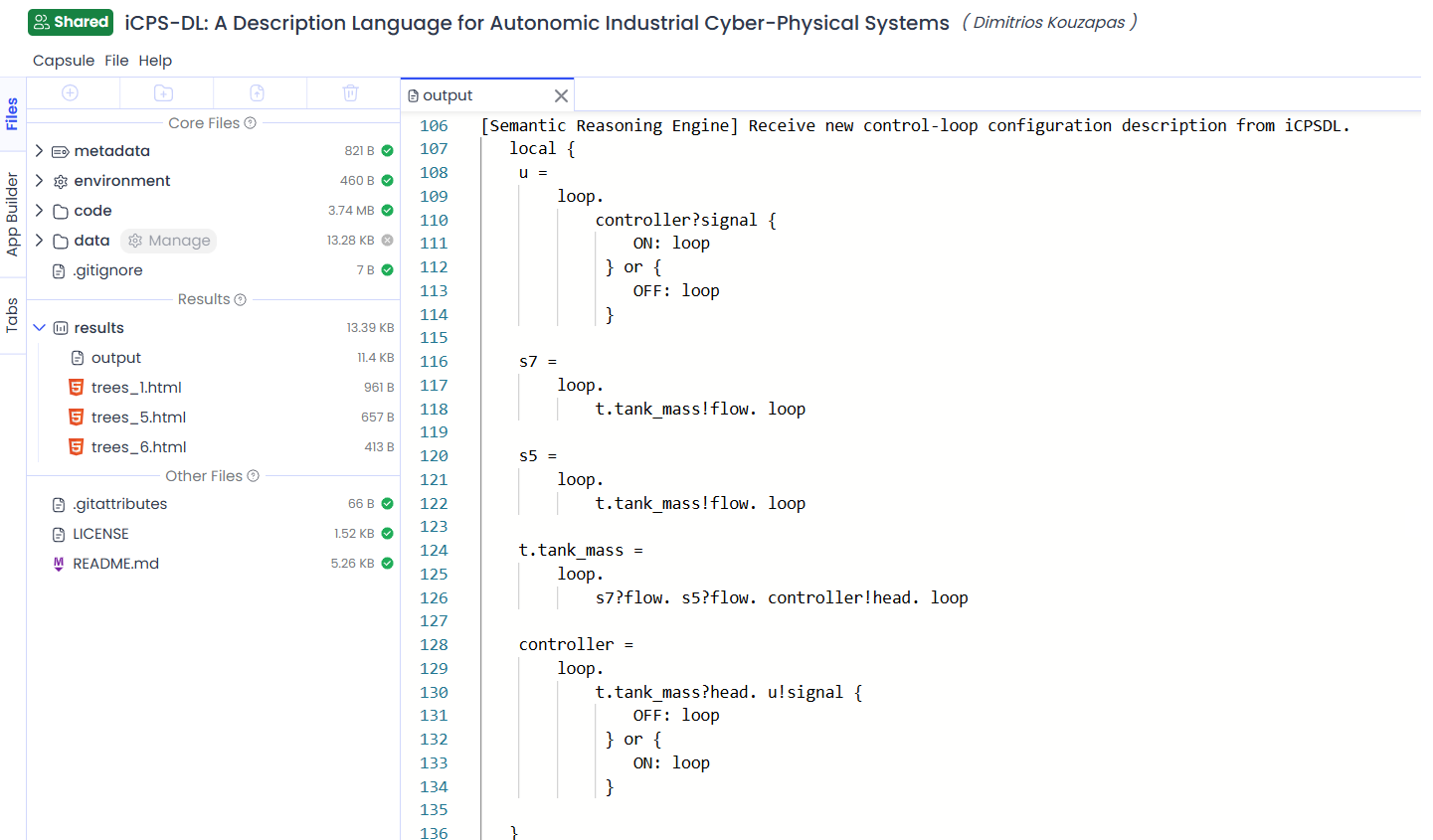}}
\end{tabular}
\end{center}
\caption{
Screenshot results of the proof-of-concept autonomic supervisor when running the CodeOcean module.
Top: Initial configuration of the running example control loop. 
Bottom: Reconfiguration of the running example control loop, after failure of sensor \sensor[6].
\label{app:fig:CodeOcean_results}
}
\end{figure*}

\begin{figure}
\begin{center}
\begin{tabular}{c}
	\fbox{\includegraphics[scale=0.4]{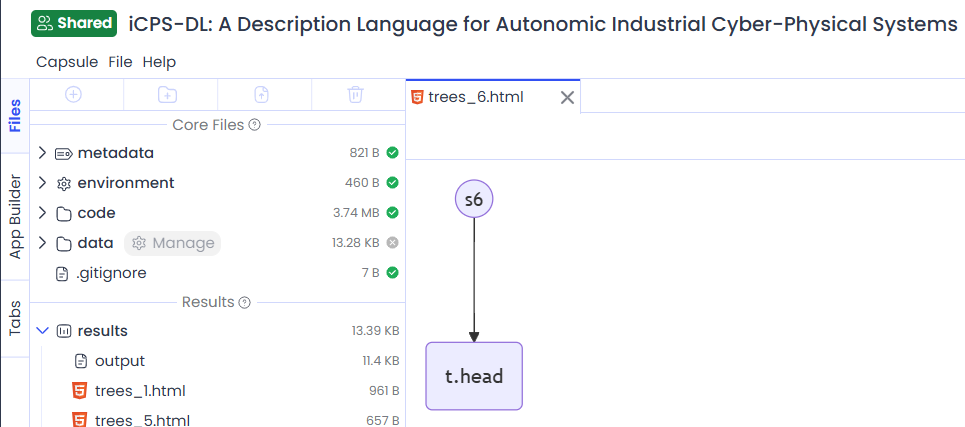}}
	\\[4mm]
	\fbox{\includegraphics[scale=0.4]{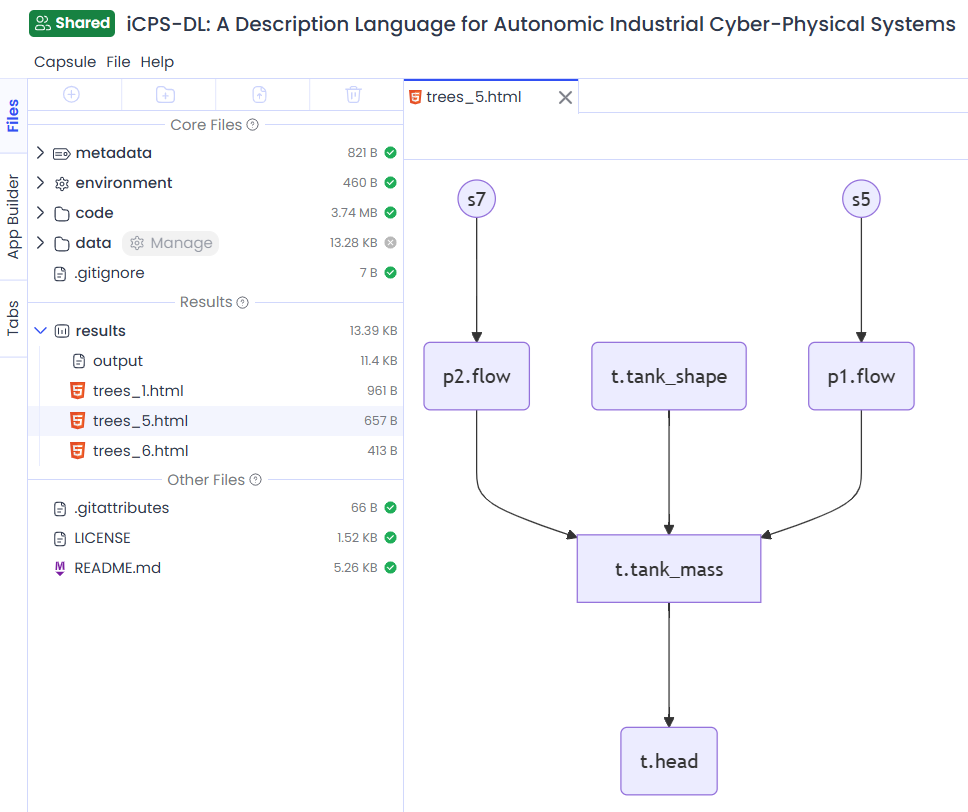}}
\end{tabular}
\end{center}
\caption{
Mermaid diagrams created by autonomic supervisor when running the CodeOcean module.
Top: Mermaid diagram of state estimator tree \treesimple[1]. 
Bottom: Mermaid diagram of state estimator tree \treesimple[2].
\label{app:fig:CodeOcean_mermaid}
}
\end{figure}

%
The next example demonstrates how the autonomic supervisor uses the \iCPSDL functionalities to monitor the running example in \secMotivatingExample.
The CodeOcean module\footnote{\url{https://codeocean.com/capsule/1441773/tree/}} also provides a proof-of-concept implementation of the autonomic supervisor monitoring a simulation of the running example.
\begin{example} [An application in WDN]
\label{ex:running_example:KB}
\rm
Fig.~\ref{app:fig:WDN_IP} (top) graphically depicts the
water distribution network of the running example in \secMotivatingExample, 
extended to include additional sensing points for physical state measurements across the network. The formal definition of the network is as:
\begin{equation*}
\WDNsimple = \lrangle{\GG, \lrangle{\Hardware, \h}, \lrangle{\AG, \modelmap}}
\end{equation*}
with $\GG = \lrangle{\VV, \EE}$.
Set, \VV, is partitioned into junctions $\VV[j] = \set{j_1, j_2, j_3}$, pipes $\VV[p] = \set{p_1, p_2}$, pumps $\VV[u] = \set{\pump}$, $\VV[t] = \set{\tank}$;  head sensing points $\Sensors[\mathsf{head}]= \set{s_1, s_3, s_6}$, and flow sensing points $\Sensors[\mathsf{flow}] = \set{s_2, s_4, s_5, s_7, s_8}$.
There are also three  cyber-physical devices, $\Hardware = \set{\dev[1], \dev[2], \dev[3]}$.

The knowledge base module of the autonomic supervisor stores the semantic description water distribution network knowledge graph, together with the \WDN domain and the agent repository.

The semantic reasoning engine can translate \WDNsimple to a state estimation graph,
$
\egraph = \translation[\GG] 
$, which is depicted in Fig.~\ref{app:fig:WDN_IP}~(middle).
The state estimation graph contains information in the form of estimator trees, each rooted at a state node. These trees dictate how cyber-physical agents can be composed to measure or estimate their root state.

For example, the reasoning engine can identify the estimation trees for estimating or measuring the $\tank.\head$ state. 
Fig.~\ref{app:fig:WDN_IP} (bottom) presents
an overlay of seven state estimator trees,
$
\set{\treesimple[1], \dots, \treesimple[7]}
$, where each tree specifies a configuration for estimating or measuring the $\tank.\head$ state. The colours blue and red highlight the estimator trees, $\treesimple[1]$ and $\treesimple[2]$, respectively, corresponding to the two control schemes described in the running example in \secMotivatingExample. The diagram also maps nodes within the trees to specific hardware devices.

A state estimator tree information enables reasoning on control scheme interactions.
For instance, consider a local configuration involving the roles of a level sensor, a controller, and a pump:
\[
\arraycolsep=2pt
\begin{array}{rcl}
\LLsimple[1]
&=& \levelSensor:   \trec{t} \out{\controller}{\headt} \tvar{t},
\\
& & \pumpp:         \trec{t} (\inp{\controller}{\ON} \tvar{t} \tor \inp{\controller}{\OFF} \tvar{t}),
\\
& & \controller:    \trec{t} \inp{\levelSensor}{\headt} (\out{\pumpp}{\ON} \tvar{t} \tor \out{\pumpp}{\OFF} \tvar{t})
\end{array}
\]
Using the agent repository, the semantic reasoning engine can construct \LLsimple[1] as an implementation of estimator tree, \treesimple[1], by mapping \sensor[6] to the level sensor role, \levelSensor. Also, it holds that $\LLsimple[1] \proves \globsimple[1]$, confirming that \LLsimple[1] is live, where:
\[
\globsimple[1] =
\rec{t} \pass{\levelSensor}{\controller}{\headt} (\pass{\controller}{\pumpp}{\ON} \var{t} \gor \pass{\controller}{\pumpp}{\OFF} \var{t})
\]

Consider the case where the event manager detects a failure in sensor \sensor[6]. The event manager updates the knowledge base by removing sensor \sensor[6]. It will then trigger the semantic reasoning to reconfigure the control loop.

Figure~\ref{app:fig:CodeOcean_results} (top) shows the CodeOcean simulation result, which generates an \iCPSDL description of \LLsimple[1]. Additionally, Figure~\ref{app:fig:CodeOcean_mermaid} (top) presents the Mermaid diagram for the state estimator tree \treesimple[1], also produced by the CodeOcean simulation. 

The semantic reasoning engine then reconstructs the state estimation graph and identifies six state estimation trees for estimating the state $\tank.\head$. 
Recall local configuration \LLsimple from Example~\ref{app:ex:local_configuration_running_example}, and
global protocol \globsimple from Example~\ref{app:ex:globsimple}.
\[
\arraycolsep=2pt
\begin{array}{rcrl}
\LLsimple   &=&     \estpp: & \trec{t} \inp{\flowSensor[1]}{\flowt} \inp{\flowSensor[2]}{\flowt} \out{\controller}{\headt} \tvar{t},
\\ &&               \flowSensor[1]: & \trec{t} \out{\estpp}{\flowt} \tvar{t},
\\ &&               \flowSensor[2]: & \trec{t} \out{\estpp}{\flowt} \tvar{t},
\\ &&               \pumpp: & \trec{t} (\inp{\controller}{\ON} \tvar{t} \tor  \inp{\controller}{\OFF} \tvar{t}),
\\ &&               \controller: & \trec{t} \inp{\estpp}{\headt} (\out{\pumpp}{\ON} \tvar{t} \tor \out{\pumpp}{\OFF} \tvar{t})
\\[2mm]
\globsimple     &=& &
\rec{t} \pass{\flowSensor[1]}{\estpp}{\flowt}
\\ & & &             \ \qquad \ \pass{\flowSensor[2]}{\estpp}{\flowt}
\\ & & &             \ \qquad \ \pass{\estpp}{\controller}{\headt} 
\\ & & &             \ \qquad \ (\pass{\controller}{\pumpp}{\ON}\var{t} \tor \pass{\controller}{\pumpp}{\OFF} \var{t})
\end{array}
\]
The semantic reasoning engine can construct \LLsimple as an implementation of estimator \treesimple[2] by associating 
\sensor[5] and \sensor[7] with the two flow sensor roles, \flowSensor[1] and \flowSensor[2], and the estimator node $\tank.\tmass$ with estimator role $\est:\localest$.
Recall also that \LLsimple is live, since $\LLsimple \proves \globsimple$.

Figure~\ref{app:fig:CodeOcean_results} (bottom) shows the CodeOcean simulation result, which generates an \iCPSDL description of \LLsimple. Additionally, Figure~\ref{app:fig:CodeOcean_mermaid} (bottom) presents the Mermaid diagram for the state estimator tree \treesimple[2], also produced by the CodeOcean simulation. 
\qed
\end{example}

\section{Full Definition of the Water Distribution Network Industrial Domain}
\label{app:sec:wdn_domain_definition}

The Water Distribution Network domain is defined by the following \iCPSDL code:
\begin{lstlisting}[language=iCPSDL]
domain {
# properties
property flow, head, tank_shape, link_shape,
signal {ON, OFF}

# estimation model
model tank_mass, junction_mass, demand_mass,
link_energy

# physical components classes
physical junction(head, flow, junction_mass):
flow -> junction_mass,
junction_mass -> flow

physical demand(head, flow, demand_mass):
flow -> demand_mass,
demand_mass -> flow

physical pipe(link_shape, flow, link_energy):
link_shape -> link_energy,
link_energy -> flow

physical tank (tank_shape, head, tank_mass):
tank_shape -> tank_mass,
tank_mass -> head

actuator pump(link_shape, flow, link_energy):
link_shape -> link_energy,
link_energy -> flow

# translation function
translation pipe -> junction:
pipe.flow -> junction.junction_mass,
junction.junction_mass -> pipe.flow,
junction.head->pipe.link_energy

translation junction -> pipe:
pipe.flow -> junction.junction_mass,
junction.junction_mass -> pipe.flow,
junction.head->pipe.link_energy

translation pipe -> tank:
pipe.flow -> tank.tank_mass,
tank.head -> pipe.link_energy

translation tank -> pipe:
pipe.flow -> tank.tank_mass,
tank.head -> pipe.link_energy

translation pump -> junction:
pump.flow -> junction.junction_mass,
junction.junction_mass -> pump.flow,
junction.head->pump.link_energy

translation junction -> pump:
pump.flow -> junction.junction_mass,
junction.junction_mass -> pump.flow,
junction.head->pump.link_energy

translation pump -> tank:
pump.flow -> tank.tank_mass

translation tank -> pump:
pump.flow -> tank.tank_mass

translation pipe -> demand:
pipe.flow -> demand.demand_mass,
demand.demand_mass -> pipe.flow,
demand.head->pipe.link_energy

translation demand -> pipe:
pipe.flow -> demand.demand_mass,
demand.demand_mass -> pipe.flow,
demand.head->pipe.link_energy

translation pump -> demand:
pump.flow -> demand.demand_mass,
demand.demand_mass -> pump.flow,
demand.head->pump.link_energy

translation demand -> pump:
pump.flow -> demand.demand_mass,
demand.demand_mass -> pump.flow,
demand.head->pump.link_energy
}
\end{lstlisting}

The agent repository used for the running example is defined by the 
following \iCPSDL code:
\begin{lstlisting}[language=iCPSDL]
repository wdn {
estimate junction_mass using
jmass = loop. producer1? flow. producer2? flow.
consumer1! flow. loop

estimate demand_mass using
dmass = loop. producer1? flow. consumer1! flow. loop

estimate tank_mass using
tmass = loop. producer1? flow. producer2? flow.
consumer1!head. loop

estimate link_energy using
lenergy = loop. producer1? head. producer2? head.
consumer1!flow. loop

sense head using headSensor = loop. consumer1! head. loop

sense flow using flowSensor = loop. consumer1! flow. loop

control pump using controller =
loop. producer1? head.
consumer1!signal { ON: loop } or { OFF: loop }

actuate pump using pumpActuator =
loop. producer1? signal { ON: loop } or { OFF: loop }
}
\end{lstlisting}

\end{document}